\documentclass[11pt]{article}

\usepackage{anysize}
\marginsize{2.3cm}{2.3cm}{1.75cm}{1.75cm}
\usepackage{hyperref}
\hypersetup{colorlinks=true,filecolor=blue,linkcolor=blue,citecolor=blue,urlcolor=blue,pdfnewwindow=true}
\usepackage{amssymb}
\usepackage{amsmath}
\usepackage{amsthm}
\usepackage{xcolor}
\usepackage{graphicx}
\usepackage{rotating}
\usepackage{algorithm}
\usepackage{multirow}
\usepackage{multicol}
\usepackage{algorithmic}
\usepackage{graphicx}
\usepackage{enumerate}
\usepackage{url}
\usepackage{hhline}

\newcommand{\D}{\displaystyle}
\newcommand{\rset}{\mathbb{R}}

\newcommand{\T}{^T}
\newcommand{\id}{I}

\newtheorem{proposition}{Proposition}

\newcommand\GC{F(p_i,p_j,d^k,d^h)}
\newcommand\UF{u}

\begin{document}

\title{Bartering integer commodities with exogenous prices}
\author{Stefano Nasini \and Jordi Castro \and Pau Fonseca i Casas\\
\and
\\
  Dept. of Statistics and Operations Research\\
  Universitat Polit\`ecnica de Catalunya\\
}
\date{January 10, 2015}
\maketitle

\abstract{The analysis of markets with indivisible goods and fixed exogenous prices has played an important role in economic models, especially in relation to wage rigidity and unemployment. This research report provides a mathematical and computational details associated to the mathematical programming based approaches proposed by Nasini et al. \cite{NasiniCastroFonseca2014} to study pure exchange economies where discrete amounts of commodities are exchanged at fixed prices. Barter processes, consisting in sequences of elementary reallocations of couple of commodities among couples of agents, are formalized as local searches converging to equilibrium allocations. A direct application of the analyzed processes in the context of computational economics is provided, along with a Java implementation of the approaches described in this research report.}

\bigskip


%\corref{cor1}
%\address{Dept. of Statistics and Operations Research\\
%Universitat Polit\`ecnica de Catalunya\\
%Jordi Girona 1--3 \\
%08034 Barcelona \\
%}

\section{Introduction} \

\normalsize
The bargaining problem concerns the allocation of a fixed quantity among a set of self-interested agents. The characterizing element of a bargaining problem is that many allocations might be simultaneously suitable for all the agents.

\newtheorem{Definition}{Definition}

\begin{Definition}
Let $\mathcal{V} \subset R^{n} $ be the space of allocations of an $n$ agents bargaining problem. Points in $\mathcal{V}$ can be compared by saying that $v^{*} \in V$ strictly dominates $v \in V$ if each component of $v^{*}$ is not less than the corresponding component of $v$ and at least one component is strictly greater, that is, $v_{i} \leq v^{*}_{i}$ for each $i$ and $ v_{i} < v^{*}_{i}$  for some $i$. This is written as $v \prec v^{*}$. Then, the Pareto frontier is the set of points of $\mathcal{V}$ that are not strictly dominated by others.

\end{Definition}

Since the very beginning of the Economic Theory \cite{Jevons1871, Edgeworth1881}, the bargaining problem has generally be adopted as the basic mathematical framework for the study of markets of excludable and rivalrous goods and a long-standing line of research focused on axiomatic approaches for the determination of a uniquely allocation, satisfying agent's interests (for details, see Nash \cite{Nash1951} and Rubinstein \cite{Rubinstein1983}).

More recently, an increasing attention has been devoted to the cases where the quantity to be allocated is not infinitesimally divisible. The technical difficulties associated to those markets have been pointed out since Shapley and Shubik \cite{ShapleyShubik1972}, who characterized the equilibria of markets where each agent can consume at most one indivisible good. After them, many authors have been studying markets with indivisible goods (see for example, Kaneko \cite{Kaneko1982}, Quinzii \cite{Quinzii1984}, Scarf \cite{Scarf1994}, and the most recent literature like Danilov et al. \cite{DanilovKoshevoyMurota2001}, Caplin and Leahy \cite{CaplinLeahy2010}). The main focus was to address the question of existence of market clearing prices in the cases of not infinitesimally divisible allocations.

Another subclass of the family of bargaining problems is associated to markets with fixed prices (for details, see Dreze \cite{Dreze1975} and Auman and Dreze \cite{AumanDreze1986}), which have played an important role in maroeconomic  models, especially on those models related to wage rigidities and unemployment. Dreze \cite{Dreze1975} described price rigidity as inequality constraints on individual prices.

Efficient algorithms to find non-dominated Pareto allocations of bargaining problems associated to markets with not infinitesimally divisible goods and fixed exogenous prices have been recently studied by Vazirani et al. \cite{Vazirani2007} and by Ozlen, Azizoglu and Burton \cite{Ozlen2011}. Our goal is to provide novel mathematical-programming based approaches to analyze barter processes, which are commonly used in everyday life by economic agents to solve bargaining problems associated to $n$-consumer-$m$-commodity markets of not infinitesimally divisible goods and fixed exogenous prices. These processes are based on \emph{elementary reallocations} (ER) of two commodities among two agents, sequentially selected from the $m(m-1)n(n-1)/4$ possible combinations. Under fixed prices, markets do not clear and the imbalance between supply and demand is resolved by some kind of quantity rationing \cite{Dreze1975}. In out analysis this quantity rationing is implicit in the process and not explicitly taken into account.

Based on this multi-agent approach, many economical systems might be simulated \cite{Wooldridge2002}, as we will see in the computational application illustrated in section 5.

Section 2 illustrates the fundamental properties of the allocation space. Section 3 provides a general mathematical programming formulation and derives an analytical expression for the Pareto frontier of the \emph{elementary reallocation problem} (ERP). It will be shown that the \emph{sequence of elementary reallocations} (SER) (the chain of ERP performed by agents along the interaction process) follows the algorithmic steps of a local search in the integer allocation space with exogenous prices. Section 4 introduce the case of network structures restricting agents interactions to be performed only among adjacent agents. In section 5 the performance of these barter processes is compared with the one of a global optimization algorithm (branch and cut).

Most of the results presented in this research reports has been studied by Nasini et al. \cite{NasiniCastroFonseca2014}.

%----------------------------------------------------------------------------------------------------------------------------------

\section{The integer allocation space with fixed prices}\

%----------------------------------------------------------------------------------------------------------------------------------

The key characteristic of an economy is: a collection $\mathcal{A}$ of $n$ agents, a collection $\mathcal{C}$ of $m$ types of commodities, a commodity space $X$ (usually represented by the nonnegative orthant in $\mathbb{R}^m$), the initial endowments $q^{i}_{j}\in X$ for $i \in \mathcal{A}$, $j \in \mathcal{C}$ (representing a budget of initial amount of commodities owned by each agent), a preference relation $\preceq_i$ on \emph{X} for each agent $i \in \mathcal{A}$.
Arrow and Debreu \cite{Debreu1983} showed that if the set $\{(x,y)\in X\times X: x \preceq_i y \}$ is closed relative to $X\times X$ the
preference relation can be represented by a real-valued function $\UF^{i}: X \longmapsto \rset$, such that, for each \emph{a} and \emph{b} belonging to \emph{X}, $\UF^{i}(a) \leqslant \UF^{i}(b) $ if and only if $a \preceq b $.

When agents attempt to simultaneously maximize their respective utilities, conditioned to balance constraints, the resulting problems are $\max \UF^{i}(\mathbf{x}) \mbox{ s.to } \sum_{i\in \mathcal{A}} x^{i}_{j} = \sum_{i\in \mathcal{A}} q^{i}_{j} \mbox{ for } j \in \mathcal{C}$, where $x^{i}_{j} \in X$, is the amount of commodity $j$ demanded by agent $i$ (from now on the superindex shall denote the agent and the subindex shall denote the commodity).

Arrow and Debreu \cite{Debreu1983} showed that under certain economic conditions (convex preferences, perfect competition and demand independence) there must be a vector of prices $\widehat{P} = (\widehat{p}_{1}, \widehat{p}_{2}, \widehat{p}_{3}, \ldots , \widehat{p}_{m})\T$, such that aggregate supplies will equal aggregate demands for every commodity in the economy.

As studied by Dreze \cite{Dreze1975}, when prices are regarded as fixed, markets do not clear and the imbalance between supply and demand is resolved by some kind of quantity rationing. The system of linear constraints associated a $n$-consumer-$m$-commodity market with fixed prices exhibits a block angular structure with rank $m+n-1$:

\begin{equation}
\label{eq:cons1}
\left[
\begin{array}{cccc}
p_{1} \, p_{2} \, \ldots \, p_{m} &  &  &  \\
 & p_{1} \, p_{2} \, \ldots \, p_{m} &  &  \\
 &  & \ddots  &\\
 &  &  & p_{1} \, p_{2} \, \ldots \, p_{m} \\
I & I & \dots & I
\end{array}
\right]
\begin{array}{c}
\mathbf{x}
\end{array}
=
\left[
\begin{array}{c}
p_{1}q^{1}_{1} + \ldots + p_{m}q^{1}_{m}\\
p_{1}q^{2}_{1} + \ldots + p_{m}q^{2}_{m}\\
\vdots\\
p_1 q^{n}_{1} + \ldots + p_{m}q^{n}_{m}\\
\mathbf{q}^{1} + \ldots + \mathbf{q}^{n}\\
\end{array}
\right],
\end{equation}
where $p_1, \dots, p_m$ are relative prices between commodities,
$\mathbf{q}^i= (q_1^i,\dots,q_m^i)\T$, and $\mathbf{x}=
(x_1^1,\dots,$ $x_m^1, \dots, x_1^n,\dots, x_m^n)\T$.
The constraints matrix of (\ref{eq:cons1}) could also be written as $
\left( \begin{array}{c} I \otimes P  \\ \mathbf{1} \otimes I  \\ \end{array}
\right)$,  where $ P= (p_{1}, p_{2}, p_{3}, \ldots , p_{m})$ and $ \otimes $
is the Kronecker product between two matrices.
Note that the linking constrains (i.e., the conservation of commodities
$(\mathbf{1} \otimes I)\mathbf{x} = \mathbf{q}^{1} + \ldots + \mathbf{q}^{n}$)
are implied by the balance equations of a network flow among the agents. This
fact will be analyzed in Section 5, where we introduced costs associated to
the flow.

All the feasible allocations lay in a ($m+n-1$) dimensional hyperplane defined
by the prices (always containing at least one solution, which is represented
by the vector of initial endowments $\mathbf{q}$), and restricted to the fact that
agents are rational: $\UF^{i}(\mathbf{x}) \geq \UF^{i}(\mathbf{q})$, for
$i \in \mathcal{V}$.

%An integer bargaining problem with exogenous prices of $m$ commodities among
%$n$ agents on a complete graph is the $n$-objective optimization problem whose
%solution set is made of the nonnegative integer points satifying (\ref{eq:cons1})
%which Pareto dominate the initial endowments and are not Pareto dominated by
%any other solutions of (\ref{eq:cons1}).

Proposition \ref{prop:1} below shows that an asymptotic approximation of an
upper bound of the number of nonnegative solutions of (\ref{eq:cons1}) is
$\mathcal{O}( \frac{n^{(mb)}}{b^m})$, where $b$ is the average amount of each
commodity, i.e., $b = \frac{\sum_{j=1}^{m}(\sum_{h=1}^{n}v^{h}_{j})}{m}$.

\begin{proposition}
\label{prop:1}
Let $\Lambda$ be the set of nonnegative solutions of (\ref{eq:cons1}),
i.e., the allocation space of a problem of bargaining integer amounts of $m$
commodities among $n$ agents with fixed prices. If the allocation space
satisfies the mild conditions $b_j= \sum_{h=1}^{n} v^{h}_{j} \geq n$ and $b_j \in \mathcal{O}(n)$, $j
= 1, \ldots,m$ (where $b_j$
is the overall amount of commodity $j$ in the system), then
$|\Lambda| \in \mathcal{O}(\frac{n^{(mb)}}{b^m})$.\\

\end{proposition}
\begin{proof}
\label{proof:1}

The set of nonnegative solutions of (\ref{eq:cons1}) is a subset of the union of bounded
sets, as
$\Lambda \subset \bigcup_{j=1}^{m} \{(x^{1}_{j} \ldots x^{n}_{j})$ $\in
\mathbb{R}^{n} : x^{1}_{j} + \ldots + x^{n}_{j}\; =\;  v^{1}_{j} + \ldots +
v^{n}_{j}; x^{1}_{j} \ldots x^{n}_{j} \geq 0 \}$. Therefore, $\Lambda$ is a
finite set, as it is the intersection between $\mathbb{Z}$ and a bounded
subset of $\mathbb{R}^{mn}$. Let $\Lambda'$ be the set of nonnegative
solutions of (\ref{eq:cons1}), without considering the price
constraints, i.e., the $n$ diagonal blocks $p_{1}x_{1}^{h} + p_{2}x_{2}^{h} +
\ldots + p_{m}x_{m}^{h} = p_{1}v_{1}^{h} + p_{2}v_{2}^{h} + \ldots +
p_{m}v_{m}^{h}$, for $h = 1, \ldots, n $. We know that $|\Lambda'| \geq
|\Lambda|$. However, $|\Lambda'|$ can be easily calculated, as the number of
solutions of $m$ independent Diophantine equations with unitary
coefficients. The number of nonnegative integer solutions of any equation of
the form $\sum_{h=1}^{n}x_j^{h} = b_j, j=1,\dots,m,$ might be seen as the number of distributions
of $b_j$ balls among $m$ boxes:  $\frac{(n + b_j-1)!}{(n-1)! \,
  b_j!}$. Since we have $m$ independent Diophantine equations of this
form,
 then the number of
possible solutions for all of them is  $\prod_{j=1}^{m} \frac{(n +
  b_{j}-1)!}{(n-1)! \, b_{j}!}$. Thus, we know that
$|\Lambda| \leq \prod_{j=1}^{m} \frac{(n + b_{j}-1)(n + b_{j}-2) \ldots
  n}{b_{j}!} \leq \prod_{j=1}^{m} \frac{(n + b_{j}-1)^{b_j}}{b_{j}!} \leq
\frac{\prod_{j=1}^{m} (n + b_{j}-1)^{b_j}}{b^m} $, where the last inequality
holds because $b_j \ge n\ge 2$. Since $b_j \in \mathcal{O}(n)$
we have that $\frac{\prod_{j=1}^{m}
  (n + b_{j}-1)^{b_j}}{b^m} \in \mathcal{O}(\frac{n^{(mb)}}{b^m})$. Hence,
$|\Lambda| \in \mathcal{O}(\frac{n^{(mb)}}{b^m})$.

\end{proof}

The set of nonnegative solutions of (\ref{eq:cons1}) represent the allocation space associated to a market with fixed prices where the quantity to be allocated is not infinitesimally divisible. The technical difficulties associated to those markets have been pointed out since Shapley and Shubik \cite{ShapleyShubik1972}, who characterized the equilibria of markets where each agent can consume at most one indivisible good. After them, many authors have been studying markets with indivisible goods (see for example, Kaneko \cite{Kaneko1982}, Quinzii \cite{Quinzii1984}, Scarf \cite{Scarf1994}, and the most recent literature like Danilov et al. \cite{DanilovKoshevoyMurota2001}, Caplin and Leahy \cite{CaplinLeahy2010}). The main focus was to address the question of existence of market clearing prices in the cases of not infinitesimally divisible allocations.

Another subclass of the family of bargaining problems is associated to markets with fixed prices (for details, see Dreze \cite{Dreze1975} and Auman and Dreze \cite{AumanDreze1986}), which have played an important role in maroeconomic  models, especially on those models related to wage rigidities and unemployment. Under fixed prices, markets do not clear and the imbalance between supply and demand is resolved by some kind of quantity rationing \cite{Dreze1975}. In out analysis this quantity rationing is implicit in the process and not explicitly taken into account.

We now set the problem of bargaining integer amounts of $m$ commodities among $n$ agents with fixed prices in a general mathematical programming framework. The aim is to construct a local search in the allocation space, based on as a sequence of elementary reallocations.

As previously seen, the linear system characterizing the space of possible allocations is (\ref{eq:cons1}). Here the conservation of commodity (i.e., the overall amount of commodity of each type must be preserved) is generalized to include arbitrary weights in the last $m$ rows of (\ref{eq:cons1}). Based on this observation consider, Nasini et al. \cite{NasiniCastroFonseca2014} proposed the following multi-objective integer non-linear optimization problem (MINOP):
\begin{subequations}\label{eq:fixed_prices_MINOP}
\begin{equation}
\hspace{-4cm} \max \quad [\UF^{i}(\mathbf{x}),~ i = 1, \ldots, n]\\
\end{equation}
\hspace{3.7cm} s. to\\
\begin{equation}
\label{eq:cons3}
\left[
\begin{array}{cccc}
P &  &  &  \\
 & P &  &  \\
 &  & \ddots  &\\
 &  &  & P \\
d^{1}I & d^{2}I & \dots & d^{n}I
\end{array}
\right]
\begin{array}{c}
\mathbf{x}
\end{array}
=
\left[
\begin{array}{c}
b^{1}\\
b^{2}\\
\vdots\\
b^{n}\\
\mathbf{b}^{0}\\
\end{array}
\right]\\
\end{equation}
\begin{equation}
\label{eq:cons4}
\begin{array}{l}
\hspace{-16mm} \UF^{i}(\mathbf{x}) \geq \UF^{i}(\mathbf{q}) \quad i = 1 \ldots, n \\
\hspace{-16mm} \mathbf{x} \in \mathbb{Z}^{mn} \geq 0,\\
\end{array}
\end{equation}
\end{subequations}
where $\UF^i:\rset^{mn}\to \rset$, $P \in \mathbb{Q}^{1 \times m}$, $ d^{i} \in \mathbb{Q} $, $ b^i \in \mathbb{Q}$, $i=1,\dots,n$, and $ \mathbf{b}^0 \in \mathbb{Q}^{m}$. The conditions $\UF^{i}(\mathbf{x}) \geq \UF^{i}(\mathbf{q}), \quad i = 1 \ldots, n$, guarantee that no agent gets worse under a feasible reallocation, which is known in general bargaining literature as the \emph{disagreement point}. The constraint matrix has a primal block-angular structure with $n$ identical diagonal blocks involving $m$ decision variables. Problem (\ref{eq:cons1}) is a particular case of (\ref{eq:fixed_prices_MINOP}) for $ d_{i} = 1, i = 1, \ldots, n$.

From a multi-objective optimization point of view, a suitable technique to generate the Pareto frontier of (\ref{eq:fixed_prices_MINOP}) is the $\varepsilon$-constraint method, which is based on converting all but one objectives into constraints. By varying the lower bounds of these constraints the exact Pareto front can theoretically be generated. This multi-objective optimization technique was proposed by Haimes, Lasdon and Wismert \cite{Haimes1971}. Recently, efficient algorithms to find non-dominated Pareto allocations of bargaining problems associated to markets with not infinitesimally divisible goods and fixed exogenous prices have been studied by Vazirani et al. \cite{Vazirani2007} and by Ozlen, Azizoglu and Burton \cite{Ozlen2009, Ozlen2011}, who developed a general approach to generate all nondominated objective vectors, by recursively identifying upper bounds on individual objectives using problems with fewer objectives.
%In the case of particularly difficult combinatorial constraints, the use of evolutionary algorithms has been quite popular in the last decades. A review of the wide range of approaches proposed to solve multi-objective integer optimization problems can be found in \cite{Deb2001}.

%------------------------------------------------------------------------------------------------------------------------------

\subsection{A specialized interior point method for markets with fixed prices}\
\label{Section:5.2.1}

%------------------------------------------------------------------------------------------------------------------------------

We introduce in this section a specialized interior point method to deal with the continuous relaxation of (\ref{eq:fixed_prices_MINOP}), as long as the utility functions $\UF^{h}(\mathbf{x})$, for $h=1, \ldots, n$, are concave. This method is based on the the specialized point algorithm for block-angular linear programs, introduced by Castro \cite{Castro00, Castro07}.

Consider a modified version of problem (\ref{eq:fixed_prices_MINOP}), in which the linking constraints are relaxed in the form of inequalities: $[d^1 I \quad I  \dots   \quad d^n I]\mathbf{x} + \mathbf{x}_0 = \mathbf{b}^{0}$, where $0 \leq \mathbf{x} \leq u_v$ and $0 \leq \mathbf{x}_0 \leq u_s$; the integrality constraints are relaxed, so that $\mathbf{x} \in \rset^{mn} \geq 0$ and the multi-objective utility function is replaced by the aggregated utility: $\sum_{h=1}^{n}\alpha_{h}\UF^{h}(\mathbf{x})$, where $\alpha_1, \ldots, \alpha_n$ are positive weights. The inequalities associated to the disagreement point (agents rationality) are replaced by equality constraints: $\UF^{h}(\mathbf{x}) - \UF^{h}(\mathbf{q}) - s^{h} = 0$, for $h = 1 \ldots, n$, where $0 \leq \mathbf{s} \leq u_s$ are slack variables, for $h = 1 \ldots n$.

We call this modified version of (\ref{eq:fixed_prices_MINOP}) the Modified Continuous Allocation Problem with Fixed Prices (MCAPFP). Note that when $u_v$ goes to zero and $u_s$ goes to infinity the $\mathbf{x}$ solution of the MCAPFP coincides with the one of the maximization of $\sum_{h=1}^{n}\alpha_{h}\UF^{h}(\mathbf{x})$ in the continuous relaxation of nonnegative solutions of (\ref{eq:cons1}). If we let $A \in \mathbb{Q}^{n+m \times mn+m}$ be the coefficient matrix associated to MCAPFP, the resulting $\mu$-KKT conditions \cite{Wright96} are:
\begin{equation} \label{eq:KKT_mu_bartering}
\begin{array}{rcll}
              A\mathbf{x} &=& \mathbf{b}, \\
              \UF^{h}(\mathbf{x}) - \UF^{h}(\mathbf{q})  - s^{h} &=& 0 & h = 1 \ldots, n,\\
\displaystyle A\T \mathbf{y} + \mathbf{z}_v - \mathbf{w}_v + \displaystyle \sum_{h=1}^{n} t^{h} \left[\begin{array}{c} \nabla \UF^{h}(\mathbf{x}) \\ \hline \mathbf{0} \end{array} \right] & = & \displaystyle \sum_{h=1}^{n} \alpha^{h} \left[\begin{array}{c} \nabla \UF^{h}(\mathbf{x}) \\ \hline \mathbf{0} \end{array} \right] \\
       T \mathbf{e}_s + \mathbf{z}_s - \mathbf{w}_s  &=& 0 \\
       X Z_v \mathbf{e}_v & = & \mu \mathbf{e}_v,  \\
       (U-X) W_v \mathbf{e}_v &=& \mu \mathbf{e}_v, \\
       S Z_s \mathbf{e}_s & = & \mu \mathbf{e}_s,  \\
       (U-S) W_s \mathbf{e}_s &=& \mu \mathbf{e}_s, \\
\end{array}
\end{equation}
where $\mathbf{e}_v \in \rset^{nm+m}$ and $\mathbf{e}_s \in \rset^{n}$ are a vectors of ones; $\mathbf{y}\in \rset^{m+n}$ and $\mathbf{z}_v,\mathbf{w}_v \in \rset^{nm+m}_{+\cup\{0\}}$ are the Lagrange multipliers (or dual variables) of $A \mathbf{x} = \mathbf{b}$ and $\mathbf{x}\ge0$, $\mathbf{x}\le u_v$ respectively; similarly, $\mathbf{t} = [t_1 \ldots t_n]\T \in \rset^{n}$ is the vector of Lagrangian multipliers of $\UF^{h}(\mathbf{x}) - \UF^{h}(\mathbf{q}) + s_h = 0$, for $h = 1 \ldots, n$ and $\mathbf{z}_s,\mathbf{w}_s \in \rset^{2n}_{+\cup\{0\}}$ are the Lagrange multipliers of $\mathbf{s} \ge 0 $, $\mathbf{s}\le u$ respectively. Primal variables must be inside the intervals $0 < \mathbf{x} < u_v$, $0 < \mathbf{s} < u_s$, $0 < \mathbf{x}_0 < u_v$. Matrices $X,Z_v,U_v,W_v \in \rset^{(nm + m) \times (nm + m)}$ are diagonal matrices made up of  vectors $\mathbf{x},\mathbf{z}_v,\mathbf{u}_v,\mathbf{w}_v$; matrices $S,T,Z_s,U_s,W_s \in \rset^{n \times n}$ are diagonal matrices made up of vectors $\mathbf{s},\mathbf{t},\mathbf{z}_s,\mathbf{u}_s,\mathbf{w}_s$. Matrix $T \in \rset^{n \times n}$ is diagonal with components $t_1, \ldots, t_n$.

Applying Newton method to (\ref{eq:KKT_mu_bartering}) and reducing the barrier parameter $\mu$ at each iteration, we have that the $\mathbf{x}$ solution of (\ref{eq:KKT_mu_bartering}) converge to the optimal allocation of the MCAPFP. The Newton's direction $(\Delta_{x}, \Delta_{s}, \Delta_{y}, \Delta_{t}, \Delta_{z_v}, \Delta_{z_s}, \Delta_{w_v}, \Delta_{w_s})$ is obtained by solving the following system in each iteration.
\begin{equation}\label{eq:KKT_jacobian_bartering}
\left[\begin{array}{cc|cc|cc|cc}
 A   &      &        &     &              &                          &           &                               \\
 V   & -I   &        &     &              &                          &           &                               \\ \hline
 Q   &      &  A^{T} & V\T & \multicolumn{2}{c|}{\multirow{2}{*}{I}} &  \multicolumn{2}{c}{\multirow{2}{*}{-I}}  \\
     &      &        &  -I &                                         &                                           \\ \hline
Z_v  &      &        &     &     X        &                          &          &                                \\
     & Z_s  &        &     &              &                S         &          &                                \\ \hline
-W_v &      &        &     &              &                          & U_v - X  &                                \\
     & -W_s &        &     &              &                          &          & U_s - S                        \\
\end{array}\right]
\left[\begin{array}{c}
\Delta_{x} \\
\Delta_{s} \\ \hline
\Delta_{y} \\
\Delta_{t} \\ \hline
\Delta_{z_v} \\
\Delta_{z_s} \\ \hline
\Delta_{w_v} \\
\Delta_{w_s}
\end{array}\right] =
%-\mathbf{F}(\mathbf{x}_k, \mathbf{y}_k, \mathbf{z}_k, \mathbf{w}_k)
\left[\begin{array}{c}
       \mathbf{r}_1 \\
       \mathbf{r}_2 \\ \hline
       \mathbf{r}_3 \\
       \mathbf{r}_4 \\ \hline
       \mathbf{r}_5 \\
       \mathbf{r}_6 \\ \hline
       \mathbf{r}_7 \\
       \mathbf{r}_8 \\
\end{array}\right]
\end{equation}
where the right-hand term is defined as
\begin{equation}
\begin{array}{ll}
       \mathbf{r}_1 &= A\mathbf{x} - \mathbf{b} \\
       \mathbf{r}_2 &= \left[\begin{array}{c} \UF^{1}(\mathbf{x}) - \UF^{1}(\mathbf{q}) - s^{1} \\ \vdots \\ \UF^{n}(\mathbf{x}) - \UF^{n}(\mathbf{q}) - s^{n} \end{array} \right]\\
       \mathbf{r}_3 &= A^{T} \mathbf{y} + \mathbf{z} - \mathbf{w}  - \displaystyle \sum_{h=1}^{n} (t^{h}- \alpha^h) \left[\begin{array}{c} \nabla \UF^{1}(\mathbf{x}) \\ \hline \mathbf{0} \end{array} \right]\\
       \mathbf{r}_4 &= T \mathbf{e}_s + \mathbf{z}_s - \mathbf{w}_s \\
       \mathbf{r}_5 &= X Z_v \mathbf{e}_v - \mu \mathbf{e}_v \\
       \mathbf{r}_6 &= X Z_s \mathbf{e}_s - \mu \mathbf{e}_s \\
       \mathbf{r}_7 &= (U_v-X) W_v \mathbf{e}_v - \mu \mathbf{e}_v \\
       \mathbf{r}_8 &= (U_s-X) W_s \mathbf{e}_s - \mu \mathbf{e}_s \\
\end{array}.
\end{equation}
Under the assumptions that $\D\frac{\partial\UF^{h}(\mathbf{x})}{\partial_{x_{i}^k}}=0$ for $h\not=k$ (i.e., $\UF^h$ only depends on $\mathbf{x}^h$ ), which are quite reasonable requirements for consumer utilities, then matrix $Q(\mathbf{x})$ results to be block-diagonal:
\begin{equation}\label{eq:StructB}
Q(\mathbf{x}) =
\left[\begin{array}{cccc|c}
     Q^1(\mathbf{x}) &                  &        &                 &            \\
                     & Q^2(\mathbf{x})  &        &                 &            \\
                     &                  & \ddots &                 &            \\
                     &                  &        & Q^n(\mathbf{x}) &            \\
                     \hline
                     &                  &        &                 & \mathbf{0}
\end{array}
\right],
\end{equation}
where, for each agent $h = 1, \ldots, n$ and each couple of commodities $i,j = 1, \ldots, m$, we have $Q_h(\mathbf{x}) \in \rset^{m \times m}$ to be defined as:
\begin{equation} \label{eq:pdf_right_hand_term}
Q^{h}_{ij}(\mathbf{x}) = \displaystyle (t^{h}-\alpha^{h}) \frac{\partial^2 \UF^{h}(\mathbf{x})}{\partial_{x_{j}^h}\partial_{x_{i}^h}}.
\end{equation}
Matrix $V \in \rset^{mn+m \times n}$ is also block-diagonal:
\begin{equation}\label{eq:StructB}
V =
\left[\begin{array}{cccc|cc}
     \nabla \UF^{1}(\mathbf{x})\T &                               &        &                              &   & \multirow{4}{*}{0}    \\
                                  & \nabla \UF^{2}(\mathbf{x})\T  &        &                              &       \\
                                  &                               & \ddots &                              &       \\
                                  &                               &        & \nabla \UF^{n}(\mathbf{x})\T &       \\
\end{array}
\right].
\end{equation}
By collecting variables $\Delta_{\widehat{x}} = [\Delta_{x} ~|~ \Delta_{s}]$, $\Delta_{\widehat{y}} = [\Delta_{y} ~|~ \Delta_{t}]$, $\Delta_{\widehat{z}} = [\Delta_{z_v} ~|~ \Delta_{z_s}]$ and $\Delta_{\widehat{w}} = [\Delta_{w_v} ~|~ \Delta_{w_s}]$ and performing elementary row operations, system (\ref{eq:KKT_jacobian_bartering}) might be reduced to
\begin{equation}\label{eq:KKT_jacobian_bartering_reduced}
\left[\begin{array}{ccccc}
    \widehat{A}                                 &                 &       &      \\
    \widehat{\Theta}                            & \widehat{A}^{T} &       &      \\
    \widehat{X}^{-1}\widehat{Z}                 &                 &  I    &      \\
    -(\widehat{U}-\widehat{X})^{-1}\widehat{W}  &                 &       &   I  \\
\end{array}\right]
\left[\begin{array}{c}
\Delta_{\widehat{x}} \\
\Delta_{\widehat{y}} \\
\Delta_{\widehat{z}} \\
\Delta_{\widehat{w}}
\end{array}\right] =
\left[\begin{array}{c}
       \widehat{\mathbf{r}}_1 \\
       \widehat{\mathbf{r}}_2 \\
       \widehat{\mathbf{r}}_3 \\
       \widehat{\mathbf{r}}_4 \\
\end{array}\right],
\end{equation}
where $\widehat{A} \in \rset^{2(n+m) \times (nm+n+m)}$ is defined as:
\begin{equation}
\widehat{A} = \left[\begin{array}{cc}
   A    &      \\
   V    &  -I
\end{array}
\right].
\end{equation}
and $\widehat{Z} \in \rset^{(nm+n+m) \times (nm+n+m)}$, $\widehat{X} \in \rset^{(nm+n+m) \times (nm+n+m)}$, $\widehat{W} \in \rset^{(nm+n+m) \times (nm+n+m)}$, $\widehat{U} \in \rset^{(nm+n+m) \times (nm+n+m)}$ are also defined by concatenating the corresponding diagonal matrices in (\ref{eq:KKT_jacobian_bartering}), as well as the right-hand term: $\widehat{\mathbf{r}}_4 = -(\widehat{U}-\widehat{X})^{-1}[\mathbf{r}_7 ~|~ \mathbf{r}_8]$,  $\widehat{\mathbf{r}}_3 = \widehat{X}^{-1}[\mathbf{r}_5 ~|~ \mathbf{r}_6]$, $\widehat{\mathbf{r}}_2 = [\mathbf{r}_3 ~|~ \mathbf{r}_4] - \widehat{\mathbf{r}}_3 + \widehat{\mathbf{r}}_4$, $\widehat{\mathbf{r}}_1 = [\mathbf{r}_1 ~|~ \mathbf{r}_2]$. Thus, variables $\Delta_{z_v}, \Delta_{z_s}, \Delta_{w_v}$ and $\Delta_{w_s}$ might be eliminated after solving the indefinite augmented form:
\begin{equation}
\left[\begin{array}{cccc}
        \widehat{A}      &          \\
      \widehat{\Theta}   & \widehat{A}^{T}
\end{array}\right]
\left[\begin{array}{c}
\Delta_{\widehat{x}} \\
\Delta_{\widehat{y}}
\end{array}\right] =
\left[\begin{array}{c}
       \widehat{\mathbf{r}}_1 \\
       \widehat{\mathbf{r}}_2
\end{array}\right],
\end{equation}
Matrix $\widehat{\Theta} \in \rset^{(nm + m + n) \times (nm + m + n)}$ is
\begin{equation}
\widehat{\Theta} = \left[\begin{array}{cc}
   \Theta^{x}  &      \\
                &  \Theta^{s}
\end{array}
\right] = \left[\begin{array}{cc}
   Q - X^{-1}Z_v + (U_v - X)W_v  &      \\
                                 &  - S^{-1}Z_s + (U_s - S)W_s
\end{array}
\right],
\end{equation}
where $\Theta^{x} \in \rset^{(nm + m) \times (nm + m)}$ and $\Theta^{s} \in \rset^{n \times n}$.
Multiplying by $-A\Theta^{-1}$ the last block of equations and summing it to the first one, we obtain that the coefficient matrix of the system to be solved to compute $\Delta_{y}$ is
\begin{equation}\label{eq:bartering_PRblock_iteration}
\begin{array}{rcl}
\widehat{A}\widehat{\Theta} \widehat{A}\T &=& \left[\begin{array}{cc}
   A\Theta^{x}A\T    &  A\Theta^{x}V\T    \\
   V\Theta^{x}A\T    &  V\Theta^{x}V\T - \Theta^{s}
\end{array}\right] \\ \\
&=&{\small \left[\begin{array}{ccc|c|ccc}
  P \Theta_{1}^{x} P\T &        &                        & d_1 P \Theta_{1}^{x} &  P \Theta_{1}^{x}  \nabla_{\UF^1} &        &       \\
                       & \ddots &                        &  \vdots              &                                   & \ddots &       \\
                       &        & P \Theta_n P\T         & d_n P \Theta_{n}^{x} &                                   &        & P \Theta_{n}^{x} \nabla_{\UF^1}\T   \\
                       &        &                        &                      &                                   &        &       \\
      \hline
                       &        &                        &                      &                                   &        &       \\
d_1 \Theta_{1}^{x} P\T & \ldots & d_n \Theta_{n}^{x} P\T & \Theta_0 + \displaystyle \sum_{h=1}^n d_h^2 \Theta_{h}^{x} &  d_1 \Theta_{1}^{x}  \nabla_{\UF^1}  & \ldots & d_n \Theta_{n}^{x} \nabla_{\UF^n}   \\
                       &        &                        &                      &                                    &        &                               \\
      \hline
               &        &                &            &                             &        &                               \\
\nabla_{\UF^1}\T \Theta_{1}^{x} P \T &        &                 & d_1 \nabla_{\UF^1}\T \Theta_{1}^{x} &  \nabla_{\UF^1}\T \Theta_{1}^{x} \nabla_{\UF^1} - \Theta_{1}^{s} &        &                 \\
               & \ddots &                 & \vdots         &                 & \ddots &                                      \\
               &        & \nabla_{\UF^n}\T \Theta_{n}^{x} P\T  & d_n \nabla_{\UF^n}\T \Theta_{n}^{x}  &                 &        & \nabla_{\UF^n}\T \Theta_{n}^{x} \nabla_{\UF^n}  - \Theta_{n}^{s}\\
\end{array}
\right] } \\ \\
& = &
\left[
\begin{array}{c|c|c}
      B     & C_0        & C_1          \\ \hline

      C_0\T & D_0        & D_{\nabla}\T \\ \hline
      C_1\T & D_{\nabla} & D_1
\end{array} \right],
\end{array}
\end{equation}
where
\begin{equation}
\nabla_{\UF^h} = \left[\begin{array}{c}
\displaystyle \frac{\partial \UF^{h}(\mathbf{x})}{\partial_{x_{m}^{h}}} \\
\vdots \\
\displaystyle \frac{\partial \UF^{h}(\mathbf{x})}{\partial_{x_{m}^{h}}} \\
\end{array}\right] \qquad h = 1, \ldots, n.
\end{equation}
Thus, by noting that the first $n$ components of the Newton direction $\Delta_{y}$ are associated to the block-angular constraints $\sum_{i \in \mathcal{C}} p_i x_i^h = \sum_{i \in \mathcal{C}} p_i q_i^h$, for $h = 1, \ldots, n$, whereas the second $m$ components of $\Delta_{y}$ are associated to the linking constraints $\sum_{h \in A} x_i^h = \sum_{h \in A} v_i^h$, for $i = 1, \ldots, m$, we define $\Delta_{y} = [\Delta_{y_1} \mbox{ } \Delta_{y_2}]$ and see that the system to be solved to compute $\Delta_{y}$ is
\begin{equation}
\begin{array}{rcl}
\widehat{A}\widehat{\Theta} \widehat{A}\T \Delta_{\widehat{y}} & = &
\left[\begin{array}{cc|c}
      B     & C_0        & C_1          \\

      C_0\T & D_0        & D_{\nabla}\T \\ \hline
      C_1\T & D_{\nabla} & D_1
\end{array} \right] \left[ \begin{array}{c}
      \Delta_{y_1}  \\
      \Delta_{y_2}  \\ \hline
      \Delta_{t}
\end{array} \right] = \left[ \begin{array}{c}
      g_{1}  \\
      g_{2}  \\ \hline
      g_{3}
\end{array} \right] \\ \\
 & = & \widehat{\mathbf{r}}_1 - \widehat{A}\widehat{\Theta}^{-1}\widehat{\mathbf{r}}_2,
 \end{array}
\end{equation}
so that we can sequentially solve the following two systems
\begin{equation}  \label{eq:system_bartering_delta_y2}
\left[\begin{array}{ccc}
  D_0 -  C_0\T B^{-1} C_0         &  & D_{\nabla}\T -  C_0\T B^{-1} C_1   \\
  D_{\nabla} - C_1\T B^{-1} C_0   &  & D_1 -  C_1\T B^{-1} C_1
\end{array} \right]\left[ \begin{array}{c}
      \Delta_{y_2}  \\
      \Delta_{t}
\end{array} \right] = \left( \left[ \begin{array}{c}
      g_{2}  \\
      g_{3}  \\
\end{array} \right] - \left[ \begin{array}{c}
      C_0\T  \\
      C_1\T  \\
\end{array} \right] B^{-1}g_1 \right),
\end{equation}

\begin{equation}\label{eq:system_bartering_delta_y1}
\begin{array}{c}
B \Delta_{y_1} = \left( g_1 - \left[ \begin{array}{cc}
      C_0  & C_1
\end{array} \right]\left[ \begin{array}{c}
      \Delta_{y_2}  \\
      \Delta_{t}
\end{array} \right] \right)
\end{array}
\end{equation}
System (\ref{eq:system_bartering_delta_y1}) is directly solvable, as $B \in \rset^{n \times n}$ is diagonal, so that the main computational effort is to solve (\ref{eq:system_bartering_delta_y2}). However, the structure of (\ref{eq:system_bartering_delta_y2}) might also been exploited, by noting that $D_1 -  C_1\T B^{-1} C_1 \in \rset^{n \times n}$ is a diagonal matrix and rewriting (\ref{eq:system_bartering_delta_y2}) in the form
\begin{equation}\label{eq:structure_bartering_delta_y2}
\begin{array}{r}
\left[\begin{array}{c|cccc}
    \Theta_0+ \sum_{h=1}^n d_h^2 \Upsilon_{h} & & d_1^2 \Upsilon_{1}\nabla_{\UF^1}                                  & \ldots & d_n^2 \Upsilon_{n}\nabla_{\UF^n}  \\
      & & & & \\
      \hline
      & & & & \\

      d_1^2 \nabla_{\UF^1}\T \Upsilon_{1}    & & \nabla_{\UF^1}\T \Upsilon_{1}\nabla_{\UF^1}-\Theta_{1}^{s} &        &            \\
      \vdots                                 & &                                                            & \ddots &            \\
      d_n^2 \nabla_{\UF^n}\T \Upsilon_{n}    & &                    &        & \nabla_{\UF^n}\T\Upsilon_{n}\nabla_{\UF^n}-\Theta_{n}^{s}
\end{array}\right]\left[ \begin{array}{c}
      \Delta_{y_2}  \\
      \Delta_{t}
\end{array} \right] \\ \\
= \left[ \begin{array}{c c}
      D_{\Upsilon}   & C_{\Upsilon} \\
      C_{\Upsilon}\T & B_{\Upsilon}
\end{array} \right] \left[\begin{array}{c}
      \Delta_{y_2}  \\
      \Delta_{t}
\end{array}\right] = \left[ \begin{array}{c}
      g_{2}  - C_0\T B^{-1}g_1\\
      g_{3}  - C_1\T B^{-1}g_1\\
\end{array} \right],
\end{array}
\end{equation}
where
\begin{equation}
\Upsilon_{h} = \Theta_{h}^{x} - \frac{\Theta_{h}^{x}P\T P \Theta_{h}^{x}}{P \Theta_{h}^{x} P\T}, \qquad h = 1, \ldots, n.
\end{equation}
By eliminating $\Delta_{t}$ from the first group of equations in (\ref{eq:structure_bartering_delta_y2}), we obtain
\begin{subequations}
\begin{eqnarray}
(D_{\Upsilon}- C_{\Upsilon}\T B^{-1}_{\Upsilon} C_{\Upsilon}) \Delta_{y_2} & = & g_{\Upsilon_{1}}  \label{eq:main_system_bartering_ATA}\\
B_{\Upsilon}\Delta_{t} & = & g_{\Upsilon_{2}} \label{eq:second_system_bartering_ATA}
\end{eqnarray}
\end{subequations}
where $g_{\Upsilon_{1}} = g_{2}  - C_0\T B^{-1}g_1 - g_{2}  - C_{\Upsilon}\T B^{-1}_{\Upsilon}(g_{3}  - C_1\T B^{-1}g_1)$ and $g_{\Upsilon_{2}} = g_{3}  - C_1\T B^{-1}g_1 - C_{\Upsilon}\Delta_{y_2}$. Since $B_{\Upsilon}$ is diagonal, $\Delta_{t}$ can be directly obtained, so that solving (\ref{eq:KKT_jacobian_bartering}) -- a system of size $2(n+m)$ -- reduced to the much smaller problem (\ref{eq:main_system_bartering_ATA}) -- a system of size $m$ --.

\subsection{The elementary reallocation problem}

%----------------------------------------------------------------------------------------------------------------------------------

The nice properties of the specialized interior point method cannot be exploited when dealing with indivisible goods and combinatorial algorithm might be taken into account. The aim of this section is to consider a general bartering scheme which is unambiguously applied to both discrete and continuous allocation spaces.

In everyday life, barter processes among people tends to achieve the Pareto frontier of problem (\ref{eq:fixed_prices_MINOP}) by a sequence of reallocations. We consider a process based on a sequence of two-commodity-two-agent reallocations, denoted as SER. Any step of this sequence requires the solution of a MINOP involving 4 variables and 4 constraints of problem
(\ref{eq:fixed_prices_MINOP}).

% Label each diagonal block as $1, 2, \ldots, n$, in accordance with the order
% they appear in (\ref{eq:cons3}), and also the variables of the diagonal blocks
% $h=1,\dots,n$ as $x_{i}^{h}$ for $i=1, \ldots, m$.
Let $\mathbf{q}$ be a feasible solution of
(\ref{eq:cons3}) and (\ref{eq:cons4}) and suppose we want to produce a
feasible change of 4 variables, such that 2 of them belong to the $i$th and
$j$th position of the diagonal block $h$ and the other belong to the $i$th
and $j$th position of the diagonal block $k$.

It can be easily shown that a feasibility condition of any affine change of
these 4 variables $q_{i}^{h} +  \Delta_{i}^{h}, q_{i}^{k}+ \Delta_{i}^{k},
q_{j}^{h} + \Delta_{j}^{h}, q_{j}^{k} + \Delta_{j}^{k}$ is that
$\Delta_{i}^{h}, \Delta_{i}^{k}, \Delta_{j}^{h}, \Delta_{j}^{k}$ must be an
integer solution of the following system of equations
\begin{equation}
\label{eq:cons5}
\left[
\begin{array}{cccc}
p_{i} & p_{j} & 0 & 0 \\
0 & 0 & p_{i} & p_{j} \\
d^{h} & 0 & d^{k} & 0 \\
0 & d^{h} & 0 & d^{k}\\
\end{array}
\right]
\left[
\begin{array}{c}
\Delta_{i}^{h}\\
\Delta_{j}^{h}\\
\Delta_{i}^{k}\\
\Delta_{j}^{k}\\
\end{array}
\right]
=
\left[
\begin{array}{c}
0\\
0\\
0\\
0\\
\end{array}
\right].
\end{equation}
The solution set are the integer points in the null space of the matrix of
system (\ref{eq:cons5}), which will be named $A$. $A$ is a
two-agent-two-commodity constraint matrix, and its rank is
three (just note that the first column is a linear combination of the
other three using coefficients $\alpha_2= \frac{p_i}{p_j}$, $\alpha_3=\frac{d^h}{d^k}$ and
$\alpha_4= -\frac{p_i d^h }{p_j d^k}$).
% This is easily seen by observing that the first column $A_{.1}$ can be
%written as $A_{.1}= \sum_{i=2}^4 \alpha_i A_{.i}$, for
%$\alpha_2= \frac{p_i}{p_j}$, $\alpha_3=\frac{d^h}{d^k}$ and
%$\alpha_4= -\frac{p_i d^h }{p_j d^k}$.
Therefore the null space has dimension one, and its integer solutions are
found on the line
\begin{equation}
\left[
\begin{array}{c}
\Delta_{i}^{h}\\
\Delta_{j}^{h}\\
\Delta_{i}^{k}\\
\Delta_{j}^{k}\\
\end{array}
\right]
= w
\left[
\begin{array}{r}
p_{j}d^{k}\\
-p_{i}d^{k}\\
-p_{j}d^{h}\\
p_{i}d^{h}\\
\end{array}
\right],
\end{equation}
for some $w = \alpha \GC$, where $ \alpha \in \mathbb{Z}$ and $F:\mathbb{Q}^4 \to
\mathbb{Q}$ provides a factor which transforms the null space direction in
the nonzero integer null space direction of smallest norm. We note that this factor can be computed as
$\GC= G(p_{j}d^{k},p_{i}d^{k},p_{j}d^{h},p_{i}d^{h})$, where
\begin{equation}
G(v_i=\frac{r_i}{w_i},i=1,\dots,l)=
\frac{\hbox{lcm}(w_i,i=1,\dots,l)}{\hbox{gcd}(\hbox{lcm}(w_i,i=1,\dots,l)\cdot v_i,i=1,\dots,l)},
\end{equation}
$r_i$ and $w_i$ being the numerator and denominator of $v_i$ ($w_i=1$ if $v_i$
is integer), and lcm and gcd
being, respectively, the least common multiple and greatest common
divisor functions.

Hence, given a feasible point $\mathbf{q}$, one can choose 4 variables, such that 2 of
them belong to the $i${th} and $j${th} position of a diagonal block $h$ and
the others belong to the $i${th} and $j${th} position of a diagonal block
$k$, in $m(m-1)n(n-1)/4$ ways. Each of them constitutes an ERP, whose Pareto
frontier is in $\mathbf{q} + \hbox{null}(A)$. The SER is a local search, which repeatedly explores
a neighborhood and chooses both a locally improving direction among the
$m(m-1)n(n-1)/4$ possible ERPs and a feasible step length $q = \alpha
\GC$, $ \alpha \in \mathbb{Z}$. For
problems of the form of (\ref{eq:fixed_prices_MINOP}) the SER might be written as follows:
\begin{equation}
\label{eq:cons6}
\mathbf{x}^{t+1} = \mathbf{x}^{t} + \alpha \GC \left[
\begin{array}{c}
\vdots      \\
p_{j}d^{k}  \\
\vdots      \\
- p_{i}d^{k}\\
\vdots      \\
- p_{j}d^{h}\\
\vdots      \\
p_{i}d^{h}  \\
\vdots      \\
\end{array}
\right]
\begin{array}{c}
\vdots  \\
 h,i \\
\vdots \\
 h,j \\
\vdots \\
 k,i \\
\vdots \\
 k,j \\
\vdots \\
\end{array}
= \mathbf{x}^{t} + \alpha \GC \Delta_{ij}^{kh},
\end{equation}
$t$ being the iteration counter. In shorter notation,
we write (\ref{eq:cons6}) as $\mathbf{x}^{t+1} = \mathbf{x}^{t} + \alpha
S_{ij}^{kh}$, where
\begin{equation}
\label{eq:sijkh}
S_{ij}^{kh}= \GC \Delta_{ij}^{kh}
\end{equation}
is a direction
of integer components. Since
the nonnegativity of $\mathbf{x}$ have to be kept along the iterations, then
we have that
\begin{equation}
\label{eq:alpharange}
\D
  -\frac{\max\left\{x_{i}^{h}/(p_{j}d^{k}) , x_{j}^{k}/(p_{i}d^{h})\right\}}{\GC}
\le \alpha \le
   \frac{\min\left\{x_{j}^{h}/(p_{i}d^{k}) , x_{i}^{k}/(p_{j}d^{h})\right\}}{\GC},
\end{equation}
or, equivalently,
\begin{equation}
\label{eq:qrange}
\D
  -\max\left\{x_{i}^{h}/(p_{j}d^{k}) , x_{j}^{k}/(p_{i}d^{h})\right\}
\le w \le
   \min\left\{x_{j}^{h}/(p_{i}d^{k}) , x_{i}^{k}/(p_{j}d^{h})\right\}.
\end{equation}
(The step length is forced to be nonnegative when the direction is both feasible and a descent direction; in our case the direction is only known to be feasible, and then negative step lengths are also considered.)

An important property of an elementary reallocation is that under the assumptions that $\D\frac{\partial \UF^{k}(\mathbf{x})}{\partial_{x_{i}^k}} :\rset^{mn} \to \rset $ is (i) non increasing, (ii) nonnegative and (iii) $\D\frac{\partial
  \UF^{k}(\mathbf{x})}{\partial_{x_{i}^j}}=0$ for $j\not=k$ (i.e., $\UF^k$ only depends on $\mathbf{x}^k$ ), which are quite
reasonable requirements for consumer utilities, then $u^{k}(\mathbf{x}+\alpha S_{ij}^{kh})$ is a unimodal function with respect to $\alpha$, as shown by the next proposition.

\begin{proposition}
\label{prop:2}
Under the definition of $u^{k}$ and $S_{ij}^{kh}$, for every feasible point $\mathbf{x} \in R^{mn}$, $u^{k}(\mathbf{x}+\alpha S_{ij}^{kh})$ is a unimodal function with respect to $\alpha$ in the interval defined by (\ref{eq:alpharange}).

\end{proposition}
\begin{proof}
Let us define $g(\alpha) = u^{k}(\mathbf{x}+\alpha S_{ij}^{kh})$, differentiable with respect to
$\alpha$. It will be shown that for all $\alpha$ in the interval (\ref{eq:alpharange}), and
$0 < \tau \in\rset$, $g'(\alpha) < 0$ implies $g'(\alpha+\tau) < 0$, which is
a sufficient condition for the unimodality of $g(\alpha)$.
By the chain rule, and using (\ref{eq:cons6}) and (\ref{eq:sijkh}), the derivative of $g(\alpha)$ can be
written as
\begin{equation}
\label{eq:gderiv}
\begin{array}{rl}
g'(\alpha) &= \nabla_{\mathbf{x}}u^{k}(\mathbf{x}+\alpha S_{ij}^{kh})
S_{ij}^{kh} \\
&= \D \GC \left(
  \frac{\partial \UF^{k}(\mathbf{x}+\alpha S_{ij}^{kh})}{\partial_{x_{i}^k}}
  (-p_j d^h) +
  \frac{\partial \UF^{k}(\mathbf{x}+\alpha S_{ij}^{kh})}{\partial_{x_{j}^k}} p_i d^h
\right).
\end{array}
\end{equation}
If $g'(\alpha) < 0$ then, from (\ref{eq:gderiv}) and since $\GC>0$, we have that
\begin{equation}
\label{eq:aux1}
 \frac{\partial \UF^{k}(\mathbf{x}+\alpha S_{ij}^{kh})}{\partial_{x_{i}^k}}
  p_j d^h >
 \frac{\partial \UF^{k}(\mathbf{x}+\alpha S_{ij}^{kh})}{\partial_{x_{j}^k}} p_i d^h.
\end{equation}
Since from (\ref{eq:cons6}) the component $(k,i)$ of $S_{ij}^{kh}$ is $\GC (-p_jd^h)<0$,
and $\D\frac{\partial \UF^{k}(\mathbf{x})}{\partial_{x_{i}^k}}$ is non
increasing, we have that for $\tau>0$
\begin{equation}
\label{eq:aux2}
   \frac{\partial \UF^{k}(\mathbf{x}+(\alpha+\tau) S_{ij}^{kh})}{\partial_{x_{i}^k}}
   \ge
 \frac{\partial \UF^{k}(\mathbf{x}+\alpha S_{ij}^{kh})}{\partial_{x_{i}^k}}.
\end{equation}
Similarly, since the component $(k,j)$ of $S_{ij}^{kh}$ is $\GC (p_id^h)>0$,
we have
\begin{equation}
\label{eq:aux3}
 \frac{\partial \UF^{k}(\mathbf{x}+\alpha S_{ij}^{kh})}{\partial_{x_{j}^k}}
   \ge
   \frac{\partial \UF^{k}(\mathbf{x}+(\alpha+\tau) S_{ij}^{kh})}{\partial_{x_{j}^k}}.
\end{equation}
Multiplying both sides of (\ref{eq:aux2}) and (\ref{eq:aux3}) by,
respectively, $p_jd^h$ and $p_id^h$, and connecting the resulting
inequalities with (\ref{eq:aux1}) we have that
\begin{displaymath}
 \frac{\partial \UF^{k}(\mathbf{x}+(\alpha+\tau) S_{ij}^{kh})}{\partial_{x_{i}^k}}
  p_j d^h >
 \frac{\partial \UF^{k}(\mathbf{x}+(\alpha+\tau) S_{ij}^{kh})}{\partial_{x_{j}^k}} p_i d^h,
\end{displaymath}
which proofs that $g'(\alpha+\tau) < 0.$
\end{proof}

Using Proposition \ref{prop:2} and the characterization of the space of integer solutions of (\ref{eq:cons5}), we are able to derive a closed expression of the Pareto frontier of the ERP, based on the behavior of $\UF(\mathbf{x}+\alpha S_{ij}^{kh})$ (see Corollary \ref{corol:1} below), as it is shown in this example:

{\small
\newtheorem{Example}[theorem]{Example}
\begin{Example}[]
Consider the following ERP with initial endowments $[40, \, 188, \, 142, \, 66]$.

{\small
\begin{equation}
\label{eq:example}
  \begin{array}{rl}
    \max & [2 - e^{-0.051 x_{1}^{1}} - e^{-0.011 x_{2}^{1}}, 2 - e^{-0.1 x_{1}^{2}} - e^{-0.031 x_{2}^{2}}] \\
    \\
    \hbox{s. to} \\
        &  5x^{1}_{1} + 10x^{1}_{2} \, =\,  2080\\
        &  5x^{2}_{1} + 10x^{2}_{2} \, =\,  1370\\
        \\
        &  5x^{1}_{1} + 6x^{2}_{1} \, =\,  1052\\
        &  5x^{1}_{2} + 6x^{2}_{2} \, =\,  1336\\
        \\
        & 2 - e^{-0.05 x_{1}^{1}} - e^{-0.01 x_{2}^{1}} \geq 1.68\\
        & 2 - e^{-0.1 x_{1}^{2}} - e^{-0.031 x_{2}^{2}} \geq 1.50\\
        \\
        &  x^{i}_{j} \geq 0 \; \in Z \; i = 1,2; \; j= 1,2;\\
  \end{array}
\end{equation} \\
}
The utility functions $g^1(\alpha)=u^{1}(\mathbf{x}+\alpha S_{12}^{12})$ and
$g^2(\alpha)=u^{2}(\mathbf{x}+\alpha S_{12}^{12})$ % as a function of $\alpha$
are
$$
\begin{array}{rl}
g^1(\alpha)=u^{1}(\mathbf{x}+\alpha S_{12}^{12}) &= u^{1}\left(
\left[
\begin{array}{c}
40  \\
188 \\
142 \\
66  \\
\end{array}
\right]
+ \alpha \left[
\begin{array}{c}
12   \\
- 6  \\
- 10 \\
5    \\
\end{array}
\right]\right) = 2 - e^{-0.051(40+12\alpha)} - e^{-0.011(188-6\alpha)}\\
\\
g^2(\alpha)=u^{2}(\mathbf{x}+\alpha S_{12}^{12}) &= u^{2}\left(
\left[
\begin{array}{c}
40  \\
188 \\
142 \\
66  \\
\end{array}
\right]
+ \alpha \left[
\begin{array}{c}
12   \\
- 6  \\
- 10 \\
5    \\
\end{array}
\right]\right) = 2 - e^{-0.1(142-10\alpha)} - e^{-0.031(66+5\alpha)},
\end{array}
$$
which are plotted in Figure \ref{fig:example}.
\begin{figure}
\label{fig:example}
\begin{center}
\includegraphics[scale=0.16]{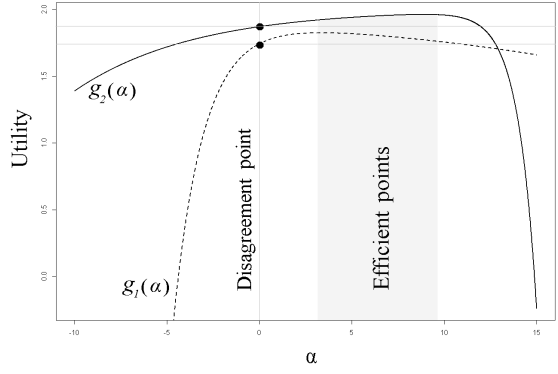} \\
\caption{\small {Plots of $g^1(\alpha)$ and $g^2(\alpha)$, and interval of $\alpha$
  associated to the Pareto frontier. The disagreement point corresponds to
  $g^1(0)$ and $g^2(0)$, the utilities in the current iterate.}}
\end{center}
\end{figure}
The continuous optimal step lengths for the two respective agents are $\textnormal{argmax}$
$g^1(\alpha) = 3.33$ and $\textnormal{argmax}$
$g^2(\alpha) = 8.94$. Due to the unimodality of
$u^{k}(\mathbf{x}+\alpha S_{ij}^{hk})$, all efficient solutions of (\ref{eq:example}) are given by
integer step lengths $\alpha \in [3.33, \,8.94]$ (see Fig. \ref{fig:example}),
i.e., for $\alpha \in \{4,5,6,7,8\}$ we have
$$
\begin{array}{lllll}
g^{1}(4) = 1.82412  &  g^1(5)= 1.81803  & g^1(6)= 1.80882 & g^1(7)= 1.79752 & g^1(8)= 1.78465, \\
g^{2}(4) = 1.93043 &  g^2(5)= 1.94035  & g^2(6)= 1.94873  & g^2(7)= 1.95558 & g^2(8)= 1.96057. \\
\end{array}
$$
Due to the unimodality of both utility functions with respect
to $\alpha$, no efficient solution exists for an $\alpha$ outside
the segment $[3.33, \, 8.94]$.

\end{Example}}
\normalsize

The above example illustrates a case where the segment between $\textnormal{argmax } u^{h} (\mathbf{x} + \alpha S_{ij}^{kh})$ and $\textnormal{argmax }u^{k} (\mathbf{x} + S_{ij}^{kh})$ contains five integer points, associated with the feasible step lengths.

The following statements give a constructive characterization of the Pareto frontier of an ERP for the case of concave utility function and linear utility functions respectively.

\newtheorem{Corollary}[]{Corollary}
\begin{Corollary}
\label{corol:1}

Let $\Gamma$ be the set of integer points in the interval $[a^{down},a^{up}]$, where $a^{down} = \min \{ \textnormal{argmax}_\alpha \UF^{k}(\mathbf{x} + \alpha S_{ij}^{kh})$, $\textnormal{argmax}_\alpha \UF^{h}(\mathbf{x} + \alpha S_{ij}^{kh})\}$ and $a^{up} = \max \{
\textnormal{argmax}_\alpha \UF^{k} (\mathbf{x} + \alpha S_{ij}^{kh})$, $\textnormal{argmax}_\alpha \UF^{h} (\mathbf{x} + \alpha S_{ij}^{kh})\}$, and let$[\alpha^{down}, \alpha^{up}]$ be the interval of feasible step lengths defined in  (\ref{eq:alpharange}). Then, due to Proposition \ref{prop:2}, the set $\mathcal{V}^*$ of Pareto efficient solutions of an ERP can be obtained as follows:

\begin{enumerate}[i.]

\item $\mathcal{V}^* = \{ [u^{h} (\mathbf{x} + \alpha S_{ij}^{kh}), u^{k} (\mathbf{x} +
\alpha S_{ij}^{kh})] \, : \; \alpha \in \Gamma \}$ if $\Gamma$ is not
empty and does not contain the zero.

\item If $\Gamma$ is empty and there exists an integer point between $0$ and
$a^{down}$ but no integer point between  $a^{up}$ and $\alpha^{up}$ then $\mathcal{V}^*$
contains the unique point given by $[u^{h} (\mathbf{x} + \alpha S_{ij}^{kh}, \;
u^{k} (\mathbf{x} + \alpha S_{ij}^{kh})] $such that $\alpha$ is the greatest
integer between 0 and $a^{down}$.
 
\item If $\Gamma$ is empty and there exists an integer point between $a^{up}$
and $\alpha^{up}$ but no integer point between $0$ and $a^{down}$ then $\mathcal{V}^*$
contains either the unique point given by $[u^{h} (\mathbf{x} + \alpha
S_{ij}^{kh}, \; u^{k} (\mathbf{x} + \alpha S_{ij}^{kh})]$ such that $\alpha$ is
the smallest integer between $a^{up}$ and $\alpha^{up}$, or $\alpha = 0$, or
both of them if they do not dominate each other. (In this case the three
possibilities must be checked, since if for only one of the utilities ---let it
be $h$, for instance--- $\UF^{h}(\mathbf{x}) > \UF^{h}(\mathbf{x}+ \bar\alpha S_{ij}^{kh})$,
$\bar\alpha$ being the smallest integer between $a^{up}$
and $\alpha^{up}$, then both values $0$ and $\bar\alpha$ are Pareto efficient.)
 
\item If $\Gamma$ is empty and there are integer points both between $a^{up}$
and $\alpha^{up}$ and between $0$ and $a^{down}$ then $V^*$ contains the points
given by $[u^{h} (\mathbf{x} + \alpha S_{ij}^{kh}, \; u^{k} (\mathbf{x} + \alpha
S_{ij}^{kh})]$ such that $\alpha$ is either the smallest integer between
$a^{up}$ and $\alpha^{up}$, or the greatest integer between 0 and $a^{down}$, or
both points if they do not dominate each other.

\item In the case that $\Gamma$ contains the zero, then no point dominates
the initial endowment $\mathbf{x}$, so that the only point in the Pareto
frontier is $\mathbf{x}$.\\
\end{enumerate}

\end{Corollary}

\begin{Corollary}\label{corol:2}

Consider the case of an economy where agents have linear utility functions with gradients $\mathbf{c}^{1}, \ldots, \mathbf{c}^{n}$ and let again $\Gamma$ be the set of integer points in the interval $[a^{down}, a^{up}]$, where $a^{down} = \min \{ \textnormal{argmax}_\alpha \alpha \mathbf{c}^{k}S_{ij}^{kh}$, $\textnormal{argmax}_\alpha \alpha \mathbf{c}^{h} S_{ij}^{kh}\}$ and $a^{up} = \max \{\textnormal{argmax}_\alpha \alpha \mathbf{c}^{k}S_{ij}^{kh}$, $\textnormal{argmax}_\alpha \alpha \mathbf{c}^{h} S_{ij}^{kh}\}$, and let $[\alpha^{down}, \alpha^{up}]$ be the interval of feasible step lengths defined in (\ref{eq:alpharange}). It might be easily seen that either $\Gamma = \mathbb{Q}$ or $\Gamma=\varnothing$. The set $\Gamma=\mathbb{Q}$ in the case $(c_{i}^{h}p_{j}d^{k} - c_{j}^{h}p_{i}d^{k})$ and $(c_{j}^{k}p_{i}d^{h} - c_{i}^{k}p_{j}d^{h})$ have opposite sighs, whereas $\Gamma=\varnothing$ if $(c_{i}^{h}p_{j}d^{k} - c_{j}^{h}p_{i}d^{k})$ and $(c_{j}^{k}p_{i}d^{h} - c_{i}^{k}p_{j}d^{h})$ have the same sign. Then, due to Proposition \ref{prop:2}, the set $\mathcal{V}^*$ of Pareto efficient solutions of an ERP may contain at most one point:

\begin{enumerate}[i.]

\item if there is at least one non-null integer between $-\max \{x_{i}^{h}/(p_{j}d^{k})$, $x_{j}^{k}/(p_{i}d^{h})\}/\GC$ and $\min \{x_{j}^{h}/(p_{i}d^{k})$, $x_{i}^{k}/(p_{j}d^{h})\}/\GC$ and $\Gamma=\varnothing$, then $\mathcal{V}^*$ only contains the unique point corresponding to the allocation $\mathbf{x}^{t+1} = \mathbf{x}^{t} + \alpha S_{ij}^{kh}$ for a step-length $\alpha$  which is either equal to $-\max \{x_{i}^{h}/(p_{j}d^{k})$ , $x_{j}^{k}/(p_{i}d^{h})\}/\GC$ (if $(c_{i}^{h}p_{j}d^{k}$ - $c_{j}^{h}p_{i}d^{k})$ and $(c_{j}^{k}p_{i}d^{h}$ - $c_{i}^{k}p_{j}d^{h})$ are negative) or for equal to $\min \{x_{j}^{h}/(p_{i}d^{k})$, $x_{i}^{k}/(p_{j}d^{h})\}/\GC$ (if $(c_{i}^{h}p_{j}d^{k}$ - $c_{j}^{h}p_{i}d^{k})$ and $(c_{j}^{k}p_{i}d^{h}$ - $c_{i}^{k}p_{j}d^{h})$ are positive).

\item $\mathcal{V}^*$ only contains the disagreement point in the opposite case.
\end{enumerate}

\end{Corollary}

Having a characterization of the Pareto frontier for any ERP in the sequence allows not just a higher efficiency in simulating the process but also the possibility of measuring the number of non dominated endowments of each of the $m(m-1)n(n-1)/4$ ERPs, which might be used as a measure of uncertainty of the process. Indeed, the uncertainty of a barter process of this type might come from different sides: i) how to choose the couple of agents and commodities in each step? ii) which Pareto efficient solution of each ERP to use to update the endowments of the system? In the next subsection we shall study different criteria for answering the first two questions.

Note that the set of non-dominated solutions of the ERP, obtained by the local search movement (\ref{eq:cons6}) might give rise to imbalances between supply and demand, as described by Dreze \cite{Dreze1975} for the continuous case. To resolve this imbalance Dreze introduce a quantity rationing, which can by also extended to the ERP.

Consider a rationing scheme for the ERP as a pair of vectors $l \in \mathbb{Z}^{m}$, $L \in \mathbb{Z}^{m}$, with $L \geq 0 \geq l$, such that the $t^{th}$ and $(t+1)^{th}$ ER verifies $l_{i} \leq \mathbf{x}_{i}^{t+1} - \mathbf{x}_{i}^{t} \leq L_{i}$, for $i = 1, \ldots, n$, where $l_i$ and $L_i$ are the $i^{th}$ components of $l$ and $L$ respectively. Thus, for two given agents $h$ and $k$ and two given commodities $i$ and $j$ we have
\begin{equation}\label{eq:alpha_RationingScheme}
l_i \leq \alpha \GC \left[
\begin{array}{c}
\vdots      \\
p_{j}d^{k}  \\
\vdots      \\
- p_{j}d^{h}\\
\vdots      \\
\end{array}
\right]
\leq L_i, \qquad
l_j \leq \alpha \GC \left[
\begin{array}{c}
\vdots      \\
-p_{i}d^{k}  \\
\vdots      \\
p_{i}d^{h}\\
\vdots      \\
\end{array}
\right]
\leq L_j,
\end{equation}
%which only requires to introduce two additional inequalities to the range of feasible step-length:
%\begin{subequations}\label{eq:alpha_range}
%\begin{equation}
%\D
%  -\frac{L_i}{\GC \min\left\{ p_{j}d^{h} , p_{i}d^{k}\right\}}
%\le \alpha \le
%  -\frac{l_i}{\GC \max\left\{ p_{j}d^{h} , p_{i}d^{k}\right\}},
%\end{equation}
%\begin{equation}
%\D
%  \frac{l_i}{\GC \min\left\{ p_{j}d^{k} , p_{i}d^{h}\right\}}
%\le \alpha \le
%  \frac{L_i}{\GC \max\left\{ p_{j}d^{k} , p_{i}d^{h}\right\}}.
%\end{equation}
%
%\end{subequations}

An open problem which has not been is not investigated by Nasini et al. \cite{NasiniCastroFonseca2014}  is the formulation of equilibrium conditions for this rationing scheme. One possibility might be the construction of two intervals for $l$ and $L$ which minimize the overall imbalances, under the conditions that (\ref{eq:alpha_RationingScheme}) is verified in each ERP, as long as $l$ and $L$ are inside the respected intervals. The integrality of the allocation space $\Lambda$ forbids a straightforward application of the equilibrium criteria proposed by Dreze \cite{Dreze1975} to the markets we are considering in this work.

%-----------------------------------------------------------------------

\subsection{Direction of movement: who exchange what?}

%-----------------------------------------------------------------------

The sequence of elementary reallocations formalized in (\ref{eq:cons5}) requires the iterative choice of couples of agents $(h,k)$
and couples of commodities $(i,j)$, i.e., directions of movement among the $m(m-1)n(n-1)/4$ in the neighborhood of the current solution. If we this choice is based on a welfare function (summarizing the utility functions of all the agents), the selection of of couples of agents and couples of commodities can be made mainly in two different ways: first improving and best improving directions of movement.

The best improving direction requires an exhaustive exploration of the neighborhood. Noting that each direction of movement in the current neighborhood constitutes a particular ERP, a welfare criterion might be the uncertainty of each elementary reallocation, measured by the number of points in the Pareto frontier of ERPs, as described in the previous subsection. A usual welfare criterion is a norm of the objective vector (e.g., Euclidean, $L_1$ or $L_\infty$ norms). Also the average marginal rate of substitution could represent an interesting criterion to select the direction of movement as a high marginal rate of substitution suggests a kind of mismatch between preferences and endowments.

If at iteration $t$ an improving direction exists the respective endowments are updated in accordance with the solution of the selected ERP: for each couple of commodities $(i,j)$ and each couple of agents $(h,k)$, agent $k$ gives $\alpha \GC p_{j}d^k$ units of $i$ to
agent $h$ and in return he/she gets $\alpha\GC p_{i}d^k$ units of $j$, for some $\alpha \in \mathbb{Z}$. At iteration $t+1$, a second couple of commodities and agents is considered in accordance with the defined criterion. If we use a first improving criterion, the process stops when the endowments keep in \emph{status quo} continuously during $m(m-1)n(n-1)/4$ explorations, i.e., when no improving direction is found in the current neighborhood.

\subsection{Linear utilities}

In microeconomic theory the utility functions are rarely linear, however the case of linear objectives appears particularly suitable from an optimization point of view and allows a remarkable reduction of operations, as the ERPs cannot have more than one Pareto-efficient solution (see Corollary \ref{corol:1}).

Consider a given direction of movement $S_{ij}^{kh}$. We know that a feasible step length $\alpha$ belongs to the interval defined by (\ref{eq:alpharange}). Since in the case of one linear objective the gradient is constant, for any direction of movement $(i,j,k,h)$ the best Pareto improvement (if there exists one) must happen in the endpoints of the feasible range of $\alpha$ (let $\alpha^{down}(i,j,k,h)$ and $\alpha^{up}(i,j,k,h)$ denote the left and right endpoints of the feasible range of $\alpha$, when the direction of movement is $(i,j,k,h)$). Therefore, the line search reduces to decide either $\alpha^{down}(i,j,k,h)$, $\alpha^{up}(i,j,k,h)$ or none of them. Then for every given point $\mathbf{x}$, we have a neighborhood of at most $m(m-1)n(n-1)/2$ candidate solutions.

Despite the idea behind the SER is a process among self-interested agents, which are by definition local optimizers, this algorithm could also be applied to any integer linear programming problem  with one linear objective: $\UF(\mathbf{x}) = c\T \mathbf{x}$. In this case however the branch and cut algorithm is much more efficient even for big instances, as we will show in
the next section.

If a first-improve method is applied, an order of commodities and agents is required when exploring the neighborhood and the equilibrium allocation might be highly affected by this order (path-dependence). The pseudocode of algorithm \ref{alg:alg2} describes the first improve search of the barter algorithm applied to the case of one linear objective function.

Note that if the nonnegativity constraints are not taken into account, problem (\ref{eq:fixed_prices_MINOP}) is unbounded for linear utility functions. This corresponds to the fact that without lower bounds the linear version of this problem would make people infinitely get into debt. As a consequence, the only possible stopping criterion, when the objective function is linear, is the fulfillment
of nonnegativity constraints, i.e. a given point $\mathbf{x}$ is a final endowment (an equilibrium of the barter process) if we have that for any direction of movement and for any given integer $\alpha$ if $c\T (\mathbf{x} + \alpha S_{ij}^{kh}) > c\T \mathbf{x}$ then $\mathbf{x} + \alpha S_{ij}^{kh}$ has some negative component. In some sense the optimality condition is now only
based on feasibility.

\begin{algorithm}
\begin{algorithmic}[1]
{ \small
\STATE Initialize the endowments $E = {<\mathbf{e}^1,\ldots,\mathbf{e}^n>}$ and utilities $U = {<u^1,\ldots,u^n>}$.
\STATE Let $t = 0$;
\STATE Let $(i,j,k,h)$ be the $t^{th}$ direction in the order set of directions;
\IF{$c\T({x} + \alpha^{down}(i,j,k,h)S_{ij}^{kh}) > c\T({x} + \alpha^{up}(i,j,k,h)S_{ij}^{kh})$ and $c\T({x} + \alpha^{down}(i,j,k,h)S_{ij}^{kh}) > c\T({x})$}
    \STATE Update the incumbent ${x} = {x} + \alpha^{down}(i,j,k,h)S_{ij}^{kh}$ and GOTO 3;
\ELSIF {$c\T({x} + \alpha^{up}(i,j,k,h)S_{ij}^{kh}) > c\T({x} + \alpha^{down}(i,j,k,h)S_{ij}^{kh})$ and $c\T({x} + \alpha^{up}(i,j,k,h)S_{ij}^{kh}) > c\T({x})$}
    \STATE Update the incumbent ${x} = {x} + \alpha^{up}(i,j,k,h)S_{ij}^{kh})$ and GOTO 3;
\ELSE
    \STATE $t = t+1$;
    \IF {$t < m(m-1)n(n-1)$}
        \STATE GOTO 4;
    \ELSE
    \STATE RETURN
\ENDIF
\ENDIF}
\end{algorithmic}
\caption{\label{alg:alg2}  First-improve SER with linear utility function}
\end{algorithm}

%---------------------------------------------------------------------------------------------------------------------------------

\subsection{The final allocation and the convergence of the SER}
\label{Subsection:6.2.1}

%---------------------------------------------------------------------------------------------------------------------------

For the case of a continuous commodity space and exogenous prices, pairwise
optimality implies global optimality, as long as all agents are initially endowed
with some positive amount of a commodity \cite{Feldman1973}. Unfortunately, the
SER described in this paper does not necessarily lead to Pareto efficient
endowments. Let $ T_{\mathbf{x}}(\alpha) = \mathbf{x} + \sum_{k\neq h}
\sum_{i\neq j} \alpha(i,j,k,h)S^{kh}_{ij}$, representing a simultaneous
reallocation of $m$ commodities among $n$ agents, with step length
$\alpha_{ij}^{kh}$ for each couple of commodities $ij$ and agents $hk$, starting
from $\mathbf{x} \in \Lambda$. Whereas a SER is required to keep feasibility
along the process, a simultaneous reallocation $ T_{\mathbf{x}}(\alpha)$ of $m$
commodities among $n$ agents does not consider the particular path and any
feasibility condition on the paths leading from $\mathbf{x}$ to $T_{\mathbf{x}}(\alpha)$.
Hence, remembering that all SERs described in this section stop when no improving
elementary reallocation exists in the current neighborhood, we can conclude
that the non existence of a feasible improving ER does not entail the non
existence of an improving simultaneous reallocation of $m$ commodities among
$n$ agents. In this sense a SER provides a lower bound of any sequence of
reallocations of more than two commodities and two agents at a time.

Consider the Lyapunov function $U(t)=\sum_{i=1}^{n}\UF^{i}(\mathbf{x}(t))$,
associating a real value to each point in the allocation space \cite{Uzawa1962}.
As $U(t)$ increases monotonically along the SER (\ref{eq:cons6}) and the
allocation space is a finite set, then $\lim_{t\rightarrow\infty}U(t)=U^{*}$.

Some understanding of the evolution of $U(t)$ along the SER iteration can
be provided.

\begin{proposition}
Consider a SER with $m$ commodities among $n$ agents with linear utility functions, i.e. $\UF^h = \mathbf{c}^h \mathbf{x}(t)$, where $c^h_i \leq 1$ (the utility functions can be rescaled by a common constant without affecting the SER). The change in the Lyapunov function from iteration $t-1$ to iteration $t$ is bounded from above by
\begin{equation}
 U(t) - U(t-1) \leq \frac{q_{\max}}{p_{\min}}\frac{d^{\max}}{d^{\min}},
\end{equation}
where $d^{\max}$ and $d^{\max}$ are the minimum and maximum elements of $d^{i} \in \mathbb{Q}$, for $i = 1 \ldots n$, as defined in (\ref{eq:fixed_prices_MINOP}); $p_{\min}$ is the minimum price and $q_{\max} = \max\{\sum_h q^h_j ~:~ j = 1 \ldots m \}$.
\end{proposition}

\begin{proof}
Let $(k,h,i,j)$ be the direction of movement selected at iteration $t$ of the SER, $\mathbf{x}(t)$ the corresponding allocation and $\delta_t = U(t) - U(t-1)$ be the change in the Lyapunov function from iteration $t-1$ to iteration $t$. In the general case we have
\begin{equation}
\delta_t = \UF^{h}(\mathbf{x}(t) + \alpha S_{ij}^{kh}) + \UF^{k}(\mathbf{x}(t) + \alpha S_{ij}^{kh}) - \UF^{h}(\mathbf{x}(t)) - \UF^{k}(\mathbf{x}(t)),
\end{equation}
which the case of linear utility functions (i.e. $\UF^h = \mathbf{c}^h \mathbf{x}(t)$)
becomes
\begin{equation}
\delta_t =  \alpha \GC \left(\left[
\begin{array}{r}
c^h_i\\
c^h_i\\
\end{array}
\right]\T \left[
\begin{array}{r}
p_{j}d^{k}\\
-p_{i}d^{k}\\
\end{array}
\right] + \left[
\begin{array}{r}
c^h_i\\
c^h_i\\
\end{array}
\right]\T \left[
\begin{array}{r}
-p_{j}d^{h}\\
p_{i}d^{h}\\
\end{array}
\right] \right),
\end{equation}
in accordance with (\ref{eq:alpharange}). Based on Corollary \ref{corol:2}, we have
\begin{subequations}
\begin{equation}
\delta_t =  - \max\left\{\displaystyle \frac{x_{i}^{h}(t)}{p_{j}d^{k}} , \frac{x_{j}^{k}(t)}{p_{i}d^{h}}\right\}\left(\left[
\begin{array}{r}
c^h_i\\
c^h_i\\
\end{array}
\right]\T \left[
\begin{array}{r}
p_{j}d^{k}\\
-p_{i}d^{k}\\
\end{array}
\right] + \left[
\begin{array}{r}
c^h_i\\
c^h_i\\
\end{array}
\right]\T \left[
\begin{array}{r}
-p_{j}d^{h}\\
p_{i}d^{h}\\
\end{array}
\right] \right),
\end{equation}
if $(c_{i}^{h}p_{j}d^{k}$ - $c_{j}^{h}p_{i}d^{k})$ and $(c_{j}^{k}p_{i}d^{h}$ - $c_{i}^{k}p_{j}d^{h})$ are negative.
\begin{equation}
\delta_t = \min\left\{ \displaystyle \frac{x_{j}^{h}(t)}{p_{i}d^{k}} , \frac{x_{i}^{k}(t)}{p_{j}d^{h}}\right\} \left(\left[
\begin{array}{r}
c^h_i\\
c^h_i\\
\end{array}
\right]\T \left[
\begin{array}{r}
p_{j}d^{k}\\
-p_{i}d^{k}\\
\end{array}
\right] + \left[
\begin{array}{r}
c^h_i\\
c^h_i\\
\end{array}
\right]\T \left[
\begin{array}{r}
-p_{j}d^{h}\\
p_{i}d^{h}\\
\end{array}
\right] \right),
\end{equation}
if $(c_{i}^{h}p_{j}d^{k}$ - $c_{j}^{h}p_{i}d^{k})$ and $(c_{j}^{k}p_{i}d^{h}$ - $c_{i}^{k}p_{j}d^{h})$ are positive.
\end{subequations}
Without lose of generality, let $p_j \leq 1$ (prices can be rescaled by choosing one commodity as a \emph{numeraire}). Then, in the economically meaningful case of having $d^h = 1$, for $h = 1 \ldots n$, we have
\begin{equation}
\delta_t \leq = \frac{q_{\max}}{p_{\min}}\frac{d^{\max}}{d^{\min}}
\end{equation}
since $q_{\max} \geq x_{i}^{h}(t)$, for all $h = 1 \ldots n$ and $i = 1 \ldots n$.
\end{proof}

In the economically meaningful case of $d^h = 1$, for all $h = 1 \ldots n$, the immediate economical interpretation of this result is that a high rage of variation of prices might result in big changes of the aggregated utility, from one bilateral exchange to another. The effect of the variability of prices on the computational performance of the SER will be studied in Subsection \ref{Subsection:6.3.2}.

\section{Bartering on networks}\

An important extension of the problem of bargaining integer amounts of $m$
commodities among $n$ agents with fixed prices is to define a network
structure such that trades among agents are allowed only for some couples of
agents who are linked in this network. In this case the conservation of
commodities $d^{1} \mathbf{x}^{1} + d^{2} \mathbf{x}^{2} + \dots + d^{n}
\mathbf{x}^{n} = d^{1} \mathbf{e}^{1} + d^{2} \mathbf{e}^{2} + \dots + d^{n}
\mathbf{e}^{n}$ is replaced by balance equations on a network, so that the
final allocation of commodity $i$ must verify ${A}\mathbf{y}_{i} =
D(\mathbf{x}_{i} - \mathbf{e}_{i})$, where $\mathbf{y}_{i}$ is the flow of
commodity $i$ in the system, ${A}$ is the incidence matrix, and $D$ is
a $n \times n$ diagonal matrix containing the weights of the conservation of
commodity $i$, that is $D =\hbox{diag}(d^{1} \dots d^{n})$ (for more details on
network flows problems see \cite{Auja1991}).

It is also possible for the final allocation to have a given maximum capacity,
that is, an upper bound of the amount of commodity $i$ that agent $h$ may
hold: $x_{i}^{h} \leq \bar{x}_{i}^{h}$.

The variables of the problem are now $x_{i}^{h}$, which again represent the
amount of commodity $i$ hold by agent $h$, $s_{i}^{h}$ which are the slack
variables for the upper bounds, and $y_{i}^{h,k}$ which are the flow of
commodity $i$ from agent $h$ to agent $k$.

The objective functions $\tilde\UF^{i}(\mathbf{x},\mathbf{y}) \, , i = 1 \ldots n$, might depend
on both the final allocation $\mathbf{x}$ and the interactions $\mathbf{y}$, since the network
topology could represent a structure of geographical proximity and
reachability.

The resulting mathematical programming formulation of the
problem of bargaining integer commodities with fixed prices among agents on a
network with upper bounds on the final allocations is as follows:

\begin{subequations}
\label{eq:cons7}
\begin{equation}
\hspace{-7cm} \max \quad [\tilde\UF^{i}(\mathbf{x},\mathbf{y}), i = 1, \ldots, n]\\
\end{equation}
\hspace{2cm} s. to\\
\begin{equation}
\left[
\begin{array}{lll|lll|lll}
P      &            &    \\
       &  \ddots    &    \\
       &            & P  \\
          \hline
\id    &            &          & \id  &         &     \\
       &  \ddots    &          &      & \ddots  &     \\
       &            & \id      &      &         & \id \\
 \hline
       &            &          &      &         &      &   A  &         &      \\
       & \mathbb{D} &          &      &         &      &      & \ddots  &      \\
       &            &          &      &         &      &      &         &   A  \\
\end{array}
\right]
\left[
\begin{array}{c}
\mathbf{x}\\
\mathbf{s}\\
\mathbf{y}\\
\end{array}
\right]
=
\left[
\begin{array}{c}
b^{1}\\
\vdots\\
b^{n}\\
\bar{x}^{1}\\
\vdots\\
\bar{x}^{n}\\
\mathbf{b}^{0}\\
\end{array}
\right]\\
\end{equation}
\begin{equation}
\begin{array}{l}
\\
\hspace{-3.2cm} \UF^{i}(\mathbf{x},\mathbf{y}) \geq \UF^{i}(\mathbf{e},\mathbf{0}) \quad i = 1 \ldots, n \\
\\
\hspace{-3.2cm} \mathbf{x} \in \mathbb{Z}^{mn} \geq 0, \qquad \mathbf{y} \in \mathbb{Z}^{mn(n-1)} \geq 0,\\\\
\end{array}
\end{equation}
\end{subequations}
where $\tilde\UF^i:\rset^{mn}\to \rset$, $P \in \mathbb{Q}^{1 \times m}$, $\mathbb{D} \in \mathbb{Q}^{mn \times mn} $, $ b^i \in
\mathbb{Q}$, $i=1,\dots,n$, $A \in \mathbb{Q}^{n \times n(n-1)}$, and $\mathbf{b}^0 \in \mathbb{Q}^{nm}$. Matrix $\mathbb{D} $ is an appropriate permutation of the diagonal matrix made of $m$ copies of the matrix $D$ with the weights of the conservation of commodity and $\tilde\UF^{i}(\mathbf{e},\mathbf{0})$ is the utility function of agent $i$ evaluated in the initial endowments $\mathbf e$ with null flow.

Problem (\ref{eq:fixed_prices_MINOP}) had $mn$ variables and $m+n$ constraints, whereas problem (\ref{eq:cons7}) has $mn(n + 1)$ variables and $n(1 + 2m)$ constraints. When a SER is applied, the definition of a network structure and the application of upper bounds to the final allocation reduce the number of feasible directions of movement in each iteration and the bound of the interval of feasible step length, as for any incumbent allocation $\mathbf   x$, the step length $\alpha$ must be such that $0 \leq \mathbf{x} + \alpha S_{ij}^{kh}
\leq \mathbf{\widetilde{x}}$.

The effect of network structures on the performance of a barter process has been previously studied by Bell \cite{Bell1998} and by Wilhite \cite{Wilhite2001}, for the case of endogenous prices and continuous commodity space. In this case the process takes into account how agents update prices each time they perform a bilateral trade. Reasonably, prices should be updated based either on the current state of the only two interacting agents or on the state of the overall population or also on the history of the system, such as previous prices. Bell showed that centralized network structures, such as a stars, exhibit a faster convergence to an equilibrium allocation.

It must be noted that any sequence of bilateral trades intrinsically gives rise to a network structure generated by the set of couples of agents interacting along the process. Such a structure might be statistically analyzed in term of its topological properties, as it is done in the next section with a battery of problems of different sizes.

%---------------------------------------------------------------------------------------------------------------------------------

\section{Applications in computational economics}\
\label{Section:6.3}

%---------------------------------------------------------------------------------------------------------------------------------

The aim of this section is to provide an inclusive application in the field of computational economics of the mathematical programming based models and methods proposed thus far. All the data sets used to replicates the results illustrated in this section can be downloaded from
\small
\begin{quotation}
\hspace{-6mm} \url{https://www.dropbox.com/sh/qekoyisyz1bzeej/AACHor8HbYU_KbYopTPxTjzca?dl=0},
\end{quotation}
\normalsize
along with a Java code implementing the previously described SERs. The reader could also modify the codes and independently use the same data to run his modified code and check his hypothesis about social bartering.

%---------------------------------------------------------------------------------------------------------------------------------

\subsection{Numerical comparison between the simultaneous reallocation and the SERs}\
\label{Subsection:6.3.1}

%---------------------------------------------------------------------------------------------------------------------------------

We first consider the number of ERs required to equilibrate the system and study their relationship with the size of the problem. In fact a numerical comparison with a global solver, such as the branch and cut, is provided to evaluate the efficiency of a decentralized barter economy in comparison with the action of a centralized global planner.

We have already seen that a SER can also be applied to any integer linear programming problem of the form (\ref{eq:fixed_prices_MINOP}), where the individual utilities are aggregated in a single welfare function. If this aggregated welfare is defined as a linear function of the endowments of the form $\UF(\mathbf x) = c\T {\mathbf x}$, the comparison of the SERs with the standard branch and cut algorithm is easily carried out.

Considering the ERP as the  basic operation of a SER and the simplex iteration as the basic operation of the branch and cut algorithm, the comparison between the two methods is numerically shown in Table \ref{tab:SER_BranchCut} for three replications of 11 problems with the same number of agents and commodities, which amounts to 33 instances. The branch and cut implementation of the state-of-the-art optimization solver Cplex was used.

\begin{table}[H]
\begin{center}
\scalebox{0.82}{
\begin{tabular}{r|r|rrr|rr|rr}
  \hline
\multirow{2}{*}{size} & \multirow{2}{*}{initial welfare} & \multicolumn{3}{|c|}{first-improve} & \multicolumn{2}{|c|}{best-improve} & \multicolumn{2}{|c}{branch and cut}\\
%\hhline{~~-------}
      &           & neighborhood &  ERPs & solution  & ERPs & solution & simplex & solution \\
  \hline
   10 &    75.134 &      0.66    &    267 &    353.269 &    91 &    365.126 &    87 &   394.630 \\
   10 &   147.958 &      0.84    &    271 &    763.188 &    91 &    767.371 &    12 &   769.861 \\
   10 & 1.205.972 &      0.77    &    375 &  3.925.921 &    74 &  3.844.165 &    70 & 4.060.685 \\

   15 &   297.713 &      0.70    &  1.343 &  1.455.839 &   215 &  1.471.387 &    49 & 1.488.149 \\
   15 &   326.996 &      0.71    &  1.090 &  2.544.271 &   237 &  2.554.755 &    63 & 2.614.435 \\
   15 &   625.800 &      0.71    &    806 &  2.640.317 &   224 &  2.644.008 &    76 & 2.684.016 \\

   20 &   183.573 &      0.67    &  2.759 &  3.432.832 &   378 &  3.425.665 &   110 & 3.525.421 \\
   20 & 1.064.023 &      0.81    &  1.582 &  4.197.757 &   361 &  4.194.187 &    94 & 4.331.940 \\
   20 &   201.377 &      0.78    &  2.629 &  1.017.906 &   351 &  1.089.860 &    80 & 1.180.977 \\

   25 &   228.365 &      0.89    &  4.358 &  2.221.790 &   648 &  2.226.152 &   237 & 2.271.552 \\
   25 &   687.492 &      0.65    &  2.806 &  3.416.982 &   572 &  3.403.937 &   113 & 3.462.043 \\
   25 &   323.495 &      0.61    &  4.706 &  2.262.657 &   666 &  2.245.817 &    50 & 2.474.429 \\

   30 &   973.955 &      0.79    &  6.648 &  5.428.473 &   975 &  5.427.207 &   101 & 5.377.843 \\
   30 & 1.811.905 &      0.82    & 13.126 &  8.945.605 & 1.084 &  8.953.611 &   127 & 9.080.651 \\
   30 & 1.302.404 &      0.85    & 12.089 &  7.583.841 &   957 &  7.573.400 &   132 & 7.605.525 \\

   35 &   653.739 &      0.87    & 13.201 &  3.456.918 & 1.310 &  3.458.570 &   112 & 3.474.126 \\
   35 &   564.905 &      0.80    &  8.772 &  3.579.713 & 1.308 &  3.585.815 &    77 & 3.599.639 \\
   35 &   753.056 &      0.83    & 14.199 &  5.132.226 & 1.290 &  5.107.933 &    67 & 5.333.123 \\

   40 &   482.570 &      0.87    & 16.307 &  2.429.707 & 1.608 &  2.428.731 &   145 &  2.446.953 \\
   40 &   430.174 &      0.68    &  7.885 &  5.281.060 & 1.640 &  5.229.740 &    90 &  5.279.631 \\
   40 & 2.795.862 &      0.79    & 14.240 & 19.175.278 & 1.578 & 14.503.963 &   186 & 19.276.444 \\

   45 & 3.392.010 &      0.98    & 62.398 & 22.681.229 & 2.300 & 22.664.443 &   162 & 22.728.195  \\
   45 &   842.645 &      0.92    & 12.900 &  6.606.875 & 2.137 &  6.642.397 &   204 &  6.755.016  \\
   45 & 1.909.859 &      0.97    & 48.688 & 15.979.841 & 2.173 & 15.865.744 &   180 & 16.071.407  \\

   50 &   839.559 &      0.93    & 20.615 & 4.822.082 & 2.105 &  4.859.830 &   137 &  4.895.655  \\
   50 &   718.282 &      0.97    & 20.744 & 3.586.560 & 2.459 &  3.588.633 &   160 &  3.610.194  \\
   50 & 1.570.652 &      0.99    & 58.165 & 18.872.864 & 2.530 & 19.018.519 &   180 & 19.069.868 \\

   55 &   351.051 &      0.98    & 20.344 & 2.761.203 & 2.935 &  2.748.862 & 1.242 & 2.799.187  \\
   55 &   413.656 &      0.96    & 26.780 & 4.566.394 & 2.922 &  4.569.975 &   336 & 4.585.475  \\
   55 &   551.355 &      0.99    & 32.053 & 5.136.295 & 3.139 &  5.135.647 &   253 & 5.157.444  \\

   60 &   468.575 &      0.99    & 27.208 & 1.941.409 & 3.568 &  1.949.786 &   271 &  1.995.930 \\
   60 &   501.366 &      0.99    & 34.323 & 5.051.429 & 3.521 &  5.051.836 &   313 &  5.067.154 \\
   60 &   575.950 &      0.98    & 43.227 & 4.751.072 & 3.589 &  4.747.097 &   273 &  4.801.179 \\
   \hline
\end{tabular}}
\caption{ \label{tab:SER_BranchCut}
\footnotesize {Numerical results of the SER and Branch and Cut for different instances of problem (\ref{eq:fixed_prices_MINOP}). The first column shows the number of agents and commodities of the problem. Columns 'ERPs' provide the number of elementary reallocations and column 'neighborhood' shows the proportion of neighborhood which has been explored. Columns 'solution' give the maximum total utility found. Column 'simplex' gives the number of simplex iterations performed by branch and cut.}}
\end{center}
\end{table}

The numerical results in Table \ref{tab:SER_BranchCut} shows $33$ problems where the number of agents and commodities is the same, as reported in the first column. For each of the $11$ different sizes $3$ replicates are computed.

The second column of Table \ref{tab:SER_BranchCut} shows the initial levels of social welfare, $c\T {\mathbf e}$. Columns \emph{solution} give the maximum utility found for the three respective methods (first-improve local search, best-improve local search, branch and cut algorithms).

The first-improve local search results in a reduced neighborhood explorations along the sequence of movements, as suggested by the values in the column \emph{neighborhood}, which show the proportion of possible combination of agents and commodities explored before moving to an improving direction (in comparison to the whole $m(m-1)n(n-1)/4$ candidate solutions).

The fourth and fifth columns of Tab. \ref{tab:SER_BranchCut}, named 'ERP', reports the number of movements, i.e. the number of ERPs for which the step-length $\alpha$ (as defined in (\ref{eq:alpharange})) has been non-null. The first-improve local search gives rise to a higher amount of ERPs, in comparison with the best-improve version. In addition, in most of the cases the best-improve search results in better allocations, as their value appear particularly close to the optimal solution (see the seventh and ninth columns of Tab. \ref{tab:SER_BranchCut}).

On the other hand, when competing with the simultaneous reallocation of all $M$ commodities among the $N$ agents, the sequence of best-improve elementary reallocations fails to reach comparatively good results in terms of number of elementary operations performed and goodness of the achieved final allocation.

The scatter plots in figures \ref{fig:Scatter1_lin} and \ref{fig:Scatter2_lin} show the relationship between the problem size (number of agents and commodities) and the elementary operations required for convergence (the ERPs for the best-improve SER and simplex iteration for the branch and cut), with the least square interpolation of algebraical curves and $R^2$ coefficient of determination.

\begin{figure}[H]
\centering
\scalebox{0.95}{
\begin{tabular}{cl}
& \\
& ERPs $ = \beta_0 + \beta_1$(size)  \\
& \\
& $\beta_0 = -946.2, \quad  \beta_1 = 69.5$ \\
& \\
& $R^2 = 0.270$ \\
& \\
& \\
\end{tabular}} \hspace{1cm}
\begin{minipage}[c]{0.53\textwidth}%
\centering
    \includegraphics[height=50mm, width=85mm]{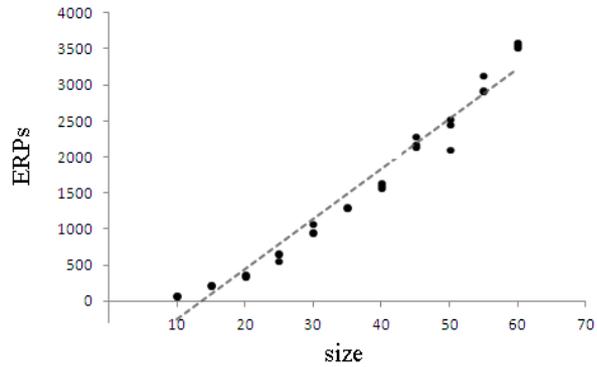}
\end{minipage}
\caption{\label{fig:Scatter1_lin} \footnotesize Scatter plot and least square approximation of a straight line through the relationship between the problem size and the number of ERPs for the best-improve SER method.}
\end{figure}

\begin{figure}[H]
\centering
\scalebox{0.95}{
\begin{tabular}{cl}
& \\
& simplex $ = \beta_0 + \beta_1$(size)  \\
& \\
& $\beta_0 = -63.4, \quad  \beta_1 =6.76$ \\
& \\
& $R^2 = 0.270$ \\
& \\
& \\
\end{tabular}} \hspace{1.8cm}
\begin{minipage}[c]{0.53\textwidth}%
\centering
    \includegraphics[height=50mm, width=85mm]{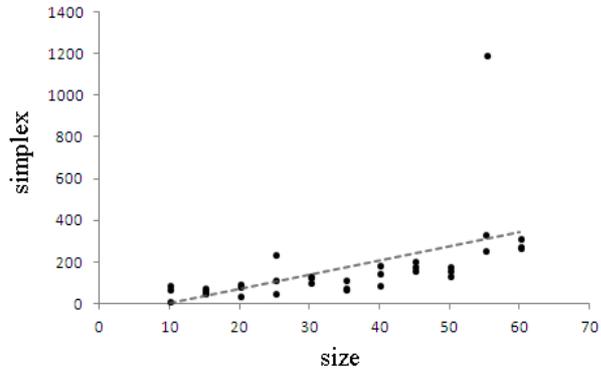}
\end{minipage}
\caption{\label{fig:Scatter2_lin} \footnotesize Scatter plot and least square approximation of a straight line through the relationship between the problem size and the number of simplex iteration for the branch and cut. }
\end{figure}

The scatter plots and least square approximation in Fig. \ref{fig:Scatter1_exp} and \ref{fig:Scatter2_exp} tries to explain the relationship between the problem size and the number of elementary operations (ERPs for the best-improve SER and the simplex pivots for the branch and cut) by an exponential curve of the form $y = \beta_0\exp(\beta_1 x)$, with the corresponding $R^2$ coefficient of determination. The same kind of plots are shown in Fig. \ref{fig:Scatter1_pow} and \ref{fig:Scatter2_pow} for the least square interpolation of a polynomial curve of the form $y = \beta_0 x^{\beta_1}$.

\begin{figure}[H]
\centering
\scalebox{0.95}{
\begin{tabular}{cl}
& \\
& ERPs $ = \beta_0\exp(\beta_1 $size)  \\
& \\
& $\beta_0 = 85.8, \quad  \beta_1 = 0.068$ \\
& \\
& $R^2 = 0.919$ \\
& \\
& \\
\end{tabular}} \hspace{1cm}
\begin{minipage}[c]{0.53\textwidth}%
\centering
    \includegraphics[height=50mm, width=85mm]{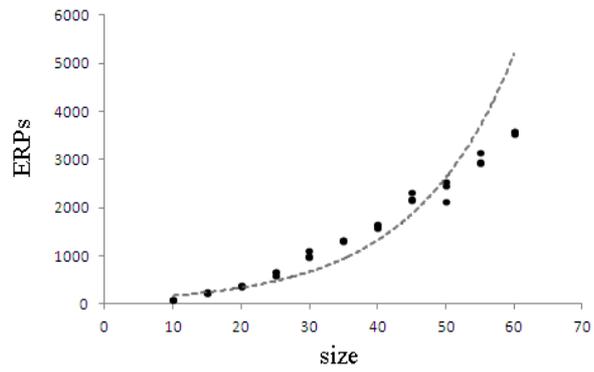}
\end{minipage}
\caption{\label{fig:Scatter1_exp} \footnotesize Scatter plot and least square approximation of an exponential curve through the relationship between the problem size and the number of ERPs for the best-improve SER method. }
\end{figure}

\begin{figure}[H]
\centering
\scalebox{0.95}{
\begin{tabular}{cl}
& \\
& simplex $ = \beta_0\exp(\beta_1 $size)   \\
& \\
& $\beta_0 = 31.4, \quad  \beta_1 = 0.039$ \\
& \\
& $R^2 = 0.602$ \\
& \\
& \\
\end{tabular}} \hspace{1cm}
\begin{minipage}[c]{0.53\textwidth}%
\centering
    \includegraphics[height=50mm, width=85mm]{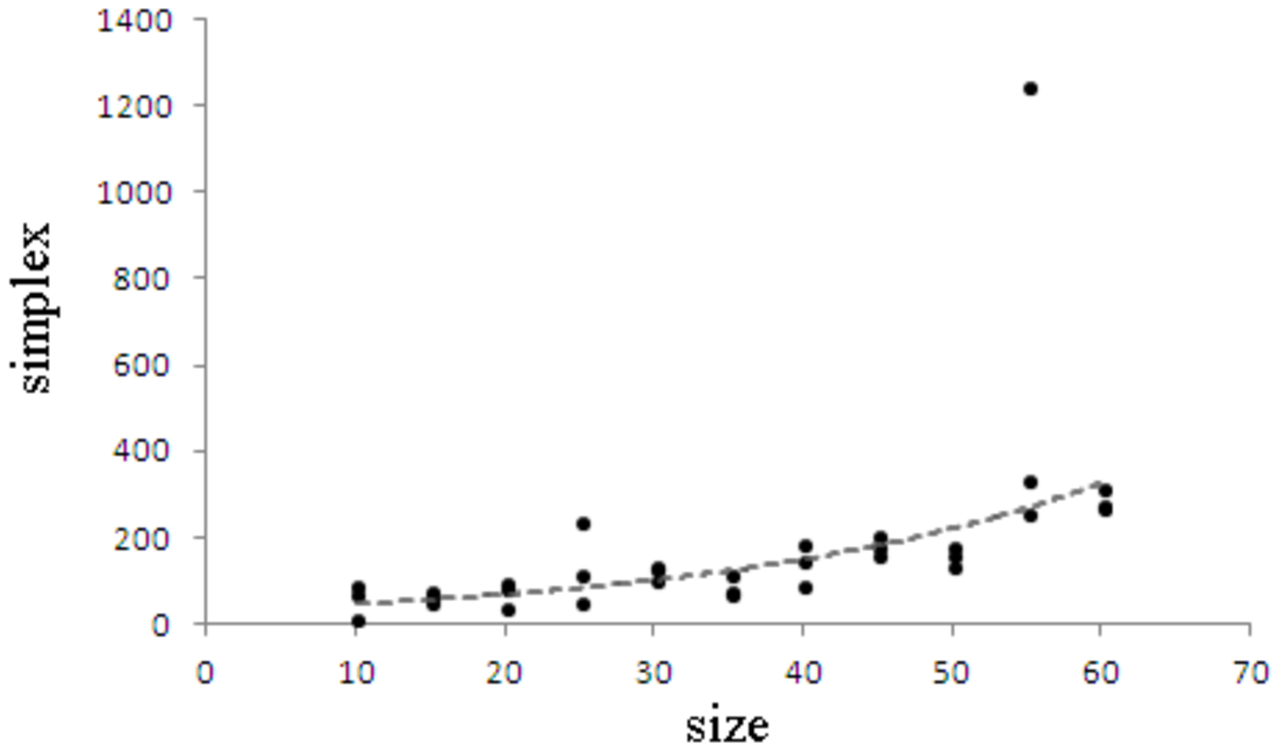}
\end{minipage}
\caption{\label{fig:Scatter2_exp} \footnotesize Scatter plot and least square approximation of an exponential curve through the relationship between the problem size and the number of simplex iteration for the branch and cut.}
\end{figure}

\begin{figure}[H]
\centering
\scalebox{0.95}{
\begin{tabular}{cl}
& \\
& ERPs $ = \beta_0$(size$)^{\beta_1}$  \\
& \\
& $\beta_0 = 0.79, \quad  \beta_1 = 2.07$ \\
& \\
& $R^2 = 0:995$ \\
& \\
& \\
\end{tabular}} \hspace{1.4cm}
\begin{minipage}[c]{0.53\textwidth}%
\centering
    \includegraphics[height=50mm, width=85mm]{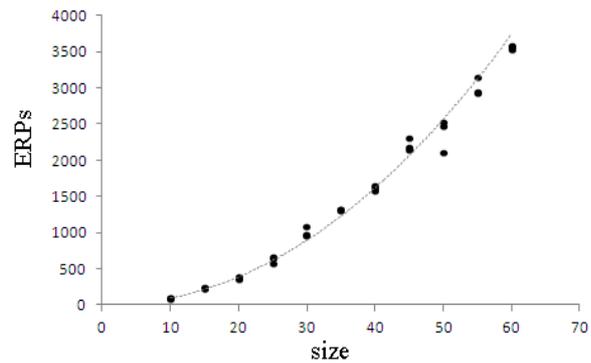}
\end{minipage}
\caption{\label{fig:Scatter1_pow} \footnotesize Scatter plot and least square approximation of a polynomial curve of the form $y = \beta_0x^{\beta_1}$ through the relationship between the problem size and the number of ERPs for the best-improve SER method. }
\end{figure}

\begin{figure}[H]
\centering
\scalebox{0.95}{
\begin{tabular}{cl}
& \\
& simplex $ = \beta_0$(size$)^{\beta_1}$  \\
& \\
& $\beta_0 = 2.84, \quad  \beta_1 = 1.100$ \\
& \\
& $R^2 = 0.567$ \\
& \\
& \\
\end{tabular}} \hspace{1.2cm}
\begin{minipage}[c]{0.53\textwidth}%
\centering
    \includegraphics[height=50mm, width=85mm]{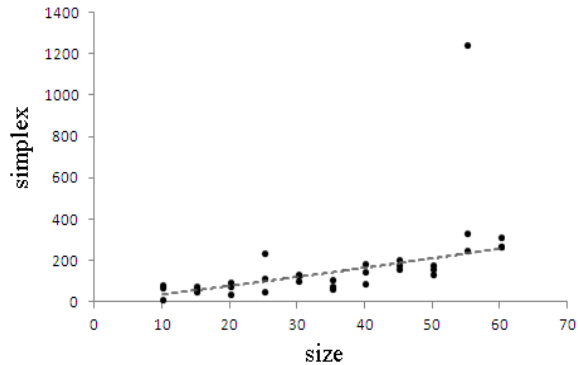}
\end{minipage}
\caption{\label{fig:Scatter2_pow} \footnotesize Scatter plot and least square approximation of a polynomial curve through the relationship between the problem size and the number of simplex iteration for the branch and cut.}
\end{figure}

This results quite clearly suggest a quadratic growth of the expected ERPs with respect to the size of the problem, in accordance with the a coefficient of determination of $0.995$. Instead the number simplex iteration of the branch and cut algorithm seems not to be well fitted by any of the proposed curves.

From the same computational view, other sequences of reallocation have been studied by Bell \cite{Bell1998}, who analyzed the performance of the process under a variety of network structures restricting the interactions to be performed only among adjacent agents. She studied a population of Cobb Douglas' agents trading continuous amount of two commodities with local Walrasian prices and focused on the speed of convergence to an equilibrium price and allocation, observing that more centralized networks converge with fewer trades and have less residual price variation than more diverse networks.

Bell relied only on the number of trades as a measure of the speed of convergence, which we regarded as movements in the local search formalizing the process. Instead, ten years ago Wilhite \cite{Wilhite2001} also toke into account the cost imposed by searching and negotiating, which we regarded as the exploration of the neighborhood\footnote{Note that in the special case of being interested in an aggregate social welfare, a system of many local optimizers (agents) could be highly inefficient if compared with a global optimizer, who acts for the 'goodness' of the system, as in the case of branch and cut. Also the increase of elementary operation of the barter algorithm is much higher than the one of the branch and cut, particularly when the direction of movement is selected in a best-improve way, as it is shown in Table \ref{tab:SER_BranchCut}. The economical interpretation suggests that if the time taken to reach an equilibrium allocation is too long, it is possible that this equilibrium is eventually never achieved in real social systems, where perturbing events (change in preferences, appetence of new types of commodities, etc.) might take place.}.

%---------------------------------------------------------------------------------------------------------------------------------

\subsection{The effect of preferences, prices, endowments}\
\label{Subsection:6.3.2}

%---------------------------------------------------------------------------------------------------------------------------------

The aim of this section is to study how the initial condition of the economy, that is to say, preferences, prices and endowments, are able to affect the computational performance of the barter processes previously defined and the emerging social structure of economical interaction.

A first question when sequences of elementary reallocations are studied might be related to the analysis of which initial condition of the system is more likely to  affect the number of non dominated allocations (improving directions), the number of negotiations (neighborhoods explored) and the emerging structure of interaction among agents.

To study the number of non dominated allocations obtained as a result of sequences of elementary reallocations, a method for the enumeration of all possible non-dominated paths from the known initial endowments is described. To do so, the $m(m-1)n(n-1)/4$ directions are explored in each step in such a way that a bundle non dominated reallocations are kept. Let $r$ be the number of non-dominated reallocation in the first iteration; for each $i = 1, \ldots, r$ a collection of $l_i \leq m(m-1)n(n-1)/4$ non-dominated directions are obtained. The bundle of non-dominated solutions are thus updated in each wave by adding and allocation in accordance with this enumerative procedure.

This procedure requires a method to find Pareto-optimal vectors each time $m(m-1)n(n-1)/4$ ERPs are solved. Corley and Moon \cite{CorleyMoon1985} proposed an algorithm to find the set $V^*$ of Pareto vectors among \emph{r} given vectors $\mathcal{V} = \{v_{1},v_{2}, \ldots, v_{r}\}$, where $ v_{i} = (v_{i1} ,v_{i2}, \ldots, v_{in}) \in \rset^{n}$, $i=1, 2, \ldots, r$. Sastry and Mohideen \cite{SastryMohideen1999} observed that the latter algorithm is incorrect and presented a modified version. In our implementation of the the best-improve barter process, we use the modified Corley and Moon algorithm of \cite{SastryMohideen1999}, shown below.

\small
\begin{description}
\item[Step 1.] Set $i=1$, $j=2$.
\item[Step 2.] If $i = r-1$, go to Step 6. For $k = 1,2, \ldots, n $, if $v_{jk}\leq v_{ik}$, then go to Step 3; else, if $v_{ik}\leq v_{jk}$, then go to Step 4. Otherwise, go to Step 5.
\item[Step 3.] Set $i = i+1$, $j=i+1$; go to Step 2.
\item[Step 4.] If $j = r$, put $v_{i} \in v \min V$ and $v_{j} = \{ \infty, \infty ,\ldots, \infty\}$; go to Step 3. Otherwise, set $v_{jk} = v_{rk}$, where $k = 1,2, \ldots, n $; set $r = r-1$ and go to Step 2.
\item[Step 5.] If $j = r$, put $v_{i} \in v \min V$; go to Step 3. Otherwise, set $j=j+1$ and go to Step 2.
\item[Step 6.] For $k = 1,2, \ldots, m $, if $v_{jk}\leq v_{ik}$, then put $v_{j} \in v \min V$ and stop; else, if $v_{ik}\leq v_{jk}$, then put $v_{i} \in v \min V$ and stop; Otherwise, put $v_{i}, v_{j} \in v \min V$ and stop.
\end{description}
\normalsize

The nice property of the modified Corley and Moon algorithm is that it doesn't necessarily compare each of the $r(r-1)/2$ couples of vectors for each of the $n$ components. This is actually what the algorithm do in the worst case, so that the complexity could be written as $\mathcal{O}(nr^{2})$, which is linear with respect of the dimension of the vectors and quadratic with respect to the number of vectors.

The pseudo-code to generate all sequences of elementary reallocations for $n$ linear agents, keeping the Pareto-improvement in each interaction, is shown in Algorithm \ref{alg:SER_multiobjective}.

\begin{algorithm}[H]
\begin{algorithmic}[1]
{\small
\STATE Initialize the endowments $E = {<\mathbf{q}^1,\ldots,\mathbf{q}^n>}$ and utilities $U = {<u^1,\ldots,u^n>}$.
\STATE Initialize the incumbent allocations $\widetilde{E}^t = \{E\}$ and the incumbent utilities $\widetilde{U}^t = \{U\}$.
\REPEAT
    \FOR{$\mathbf{v} \in \widetilde{E}^t$}
%       \STATE Let $Q_v (h\, k\, i\, j) = \{q^{+}_v(h\, k\, i\, j), q^{-}_v(h\, k\, i\, j)\}$
        \STATE Let $<S_\mathbf{x}, G_\mathbf{x}> $ be the set of movements and utilities
        $\{(\mathbf{x} + \alpha S_{ij}^{kh}$, $c'(\mathbf{x} + \alpha
        S_{ij}^{kh}))\}$ for each couple of commodities and agents $(i,j,k,h)$ and $\alpha \in \{\alpha^{down}(i,j,k,h),\alpha^{up}(i,j,k,h)\}$
    \ENDFOR
        \STATE Let $<S, G> = \bigcup_{x \in \widetilde{E}} <S_\mathbf{v}, U_\mathbf{v}>$ and $<S, G> = CorleyMoon(<S, G>)$
        \STATE Let $\widetilde{E}^{t+1} = \widetilde{E}^{t} \cup S $ and $\widetilde{U}^{t+1} = \widetilde{U}^{t} \cup G$
        \STATE Let $t = t+1$
\UNTIL{$\widetilde{E}^t = \widetilde{E}^{t-1}$}}
\end{algorithmic}
\caption{\label{alg:SER_multiobjective} Generating paths of all improving directions of movement}
\end{algorithm}

The function $CorleyMoon()$ applies the modified Corley and Moon algorithm to a set of utility vectors and allocation vectors and return the Pareto-efficient utility vectors with the associated allocations.

\begin{figure}[H]
\begin{multicols}{2}
\centering
\scalebox{0.70}{
\begin{tabular}{l | l l l l |l}
    \hline
\mbox{            }  \; \,  &\multicolumn{3}{c}{allocations} \mbox{          }  & \qquad & utilities \mbox{    } \\ \hline
                             \\
    iteration 0             &  18 3 3     & 13 4 55   & 22 2 2   & \qquad & 1422 559 1220     \\ \hline
                             \\
    iteration 1             &  21 3 0     & 10 4 58   & 22 2 2   & \qquad & 1608 574 1220     \\
                            &  18 3 3     & 11 4 57   & 24 2 0   & \qquad & 1422 569 1324     \\ \hline
                             \\
    iteration 2             &  24 0 0     & 7  7 58   & 22 2 2   & \qquad & 1800 571 1220     \\
                            &  19 5 0     & 10 4 58   & 24 0 2   & \qquad & 1480 574 1326     \\
                            &  21 3 0     & 8  6 58   & 24 0 2   & \qquad & 1608 572 1326    \\
                            &  21 3 0     & 8  4 60   & 24 2 0   & \qquad & 1608 584 1324    \\
                            &  21 3 0     & 10 6 56   & 22 0 4   & \qquad & 1422 567 1430    \\
                            &  21 0 3     & 8  7 57   & 24 2 0   & \qquad & 1614 566 1324    \\ \hline
                             \\
    iteration 3             &  21 3 0     & 8  4 60   & 24 2 0   & \qquad & 1608 584 1324    \\
                            &  22 2 0     & 7  7 58   & 24 0 2   & \qquad & 1672 571 1326     \\
                            &  24 0 0     & 5  9 58   & 24 0 2   & \qquad & 1800 569 1326     \\
                            &  24 0 0     & 5  7 60   & 24 2 0   & \qquad & 1800 581 1324     \\
                            &  24 0 0     & 7  9 56   & 22 0 4   & \qquad & 1614 564 1430     \\
                            &  19 5 0     & 8  4 60   & 26 0 0   & \qquad & 1480 584 1430     \\
                            &  19 5 0     & 10 2 60   & 24 2 0   & \qquad & 1480 586 1324     \\
                            &  21 1 2     & 8  6 58   & 24 2 0   & \qquad & 1608 582 1430     \\ \hline
                             \\
    iteration 4             &  21 3 0     & 8  4 60   & 24 2 0   & \qquad & 1608 584 1324     \\
                            &  24 0 0     & 5  7 60   & 24 2 0   & \qquad & 1800 581 1324     \\
                            &  19 5 0     & 8  4 60   & 26 0 0   & \qquad & 1480 584 1430     \\
                            &  19 5 0     & 10 2 60   & 24 2 0   & \qquad & 1480 586 1324     \\
                            &  21 1 2     & 8  6 58    & 24 2 0   & \qquad & 1608 582 1430     \\
                            &  21 3 0     & 8  6 58    & 24 0 2   & \qquad & 1800 579 1430     \\
                            &  22 0 2     & 7  9 56    & 24 0 2   & \qquad & 1672 581 1430     \\
                            &  24 0 0     & 7  7 58    & 22 2 2   & \qquad & 1736 582 1324     \\
                            &  20 2 2     & 7  7 58    & 26 0 0   & \qquad & 1672 583 1324     \\ \hline
                             \\
    iteration 5             &  21 3 0     & 8  4 60    & 24 2 0   & \qquad & 1608 584 1324    \\
                            &  24 0 0     & 5  7 60    & 24 2 0   & \qquad & 1800 581 1324    \\
                            &  19 5 0     & 8  4 60    & 26 0 0   & \qquad & 1480 584 1430    \\
                            &  19 5 0     & 10 2 60    & 24 2 0   & \qquad & 1480 586 1324     \\
                            &  21 1 2     & 8  6 58    & 24 2 0   & \qquad & 1608 582 1430    \\
                            &  21 3 0     & 8  6 58    & 24 0 2   & \qquad & 1800 579 1430    \\
                            &  22 0 2     & 7  9 56    & 24 0 2   & \qquad & 1672 581 1430    \\
                            &  24 0 0     & 7  7 58    & 22 2 2   & \qquad & 1736 582 1324    \\
                            &  20 2 2     & 7  7 58    & 26 0 0   & \qquad & 1672 583 1324    \\
                            &  21 0 3     & 8  7 57    & 24 2 0   & \qquad & 1544 583 1430    \\ \hline
\end{tabular}}
\includegraphics[scale=0.094]{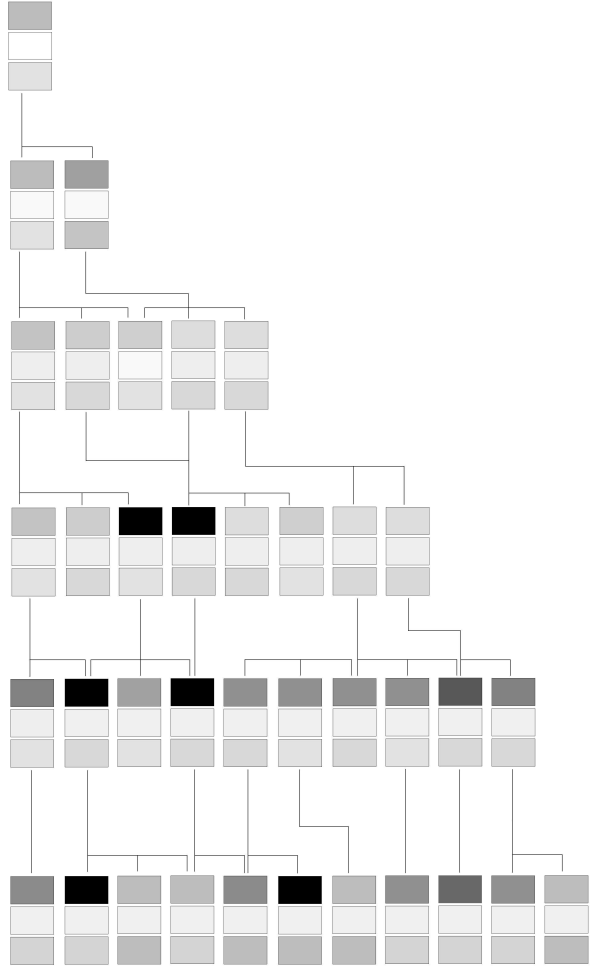}
\end{multicols}
\caption{ \label{fig:WorkedExampleMultiobjective} \footnotesize Worked example of the generation of all possible SERs, as described in algorithm \ref{alg:SER_multiobjective}.}
\end{figure}

Consider a barter process of 3 commodities among 3 agents and let the initial endowments be $\mathbf{q}^{1} = (18~ \, 3~ \, 3) $, $\mathbf{q}^{2} = (13~ \, 4~ \, 55)$ and $\mathbf{q}^{3} = (22~ \, 2~ \, 2) $. The coefficients of the linear objective functions are $\mathbf{c}^{1} = (75~ \, 11~ \, 13) $, $\mathbf{c}^{2} = (4~ \, 3~ \, 9)$ and $\mathbf{c}^{3} = (55~ \, 2~ \, 3) $. Starting from the initial solution, the sequence of two-agent-two-commodity barter leads to the movements of Figure \ref{fig:WorkedExampleMultiobjective}.

The scale of grey denotes the utility level. Starting from the initial endowments, 28 different stories of elementary reallocations might be generated, although many of them lead to the same stable allocation (local optima). We found 11 stable allocations which might be reached by some sequence of elementary allocation keeping the Pareto-optimality in each ERP.

We consider a theoretical case where $2$ agents with linear utility functions have to trade $9$ commodities. The following three factors are taken into account:
\small
\begin{itemize}
\item[-] $Fact_1$: the variability of prices;
\item[-] $Fact_2$: association between $\mathbf{q}^{1}$ and $\mathbf{c}^{1}$ and between $\mathbf{q}^{2}$ and $\mathbf{c}^{2}$ (\emph{initial stability});
\item[-] $Fact_3$: association between $\mathbf{q}^{1}$ and $\mathbf{c}^{2}$ and between $\mathbf{q}^{2}$ and $\mathbf{c}^{1}$ (\emph{dissortative matching}).
\end{itemize}
\normalsize
The aforementioned factors are measured at three levels and $4$ randomized replicates have been simulated for each combination of factors. A multivariate analysis of variance (MANOVA) is performed, considering the two following response variables
\small
\begin{itemize}
\item[-] $Resp_1$: the number of non dominated allocations;
\item[-] $Resp_2$: the number of neighborhoods explored.
\end{itemize}
\normalsize
The MANOVA results\footnote{The multiple analysis of variance is used to compare multivariate (population) means of several combinations of factors. The third and fourth columns of Table \ref{fig:MANOVA_Multiobjective_SER} report commonly used test statistics which provide a p-value assuming an $F$ distribution under the null hypothesis.} in Table \ref{fig:MANOVA_Multiobjective_SER} illustrates the effects and the significance of $Fact_3$, corresponding to the association between the initial endowments and the marginal utilities of opposite agents.  The correlation between the amounts of the initial endowments and the coefficients of the objective function of the same agent does not appear by itself to have a significant effect on the response variables.

\begin{table}[H]
\begin{center}
\scalebox{0.93}{
\begin{tabular}{l | l l l l }
\hline
                                       & df     & Pillai      & approx $F$  & p-value         \\ \hline

$Fact_1$                               & 2      & 0.098426    & 2.3033      & 0.06028         \\
$Fact_2$                               & 2      & 0.034673    & 0.7851      & 0.53624         \\
$Fact_3$                               & 2      & 0.133653    & 3.1867      & 0.01474         \\
$Fact_1 \times Fact_2$                 & 4      & 0.037110    & 0.4207      & 0.90758         \\
$Fact_1 \times Fact_3$                 & 4      & 0.070324    & 0.8109      & 0.59384         \\
$Fact_2 \times Fact_3$                 & 4      & 0.166118    & 2.0155      & 0.04701         \\ \hline
Residuals                              & 89     &             &             &                 \\ \hline
\end{tabular}}
\caption{\label{fig:MANOVA_Multiobjective_SER} \footnotesize{MANOVA analysis of the paths of all improving directions}}
\end{center}
\end{table}

The graphical illustration in Figure \ref{fig:Multiobjective_SER} supports the MANOVA results, by showing the values of the two response variables for each level of the factors. The price variability seems to have a non-linear effect to both response variables (left panel). The association between the initial endowment and the marginal utility of the same agent doesn't seem to produce a consistent change in the number of neighborhoods explored (red line in the central panel), though it does have a clear average linear effect on the number of non-dominated allocations. Differently, the correlation between the initial endowment of an agent and the coefficients of the utility function of the other exhibits negative associations with the two response variables.

\begin{figure}[H]
\begin{center}
\includegraphics[height=100mm, width=145mm]{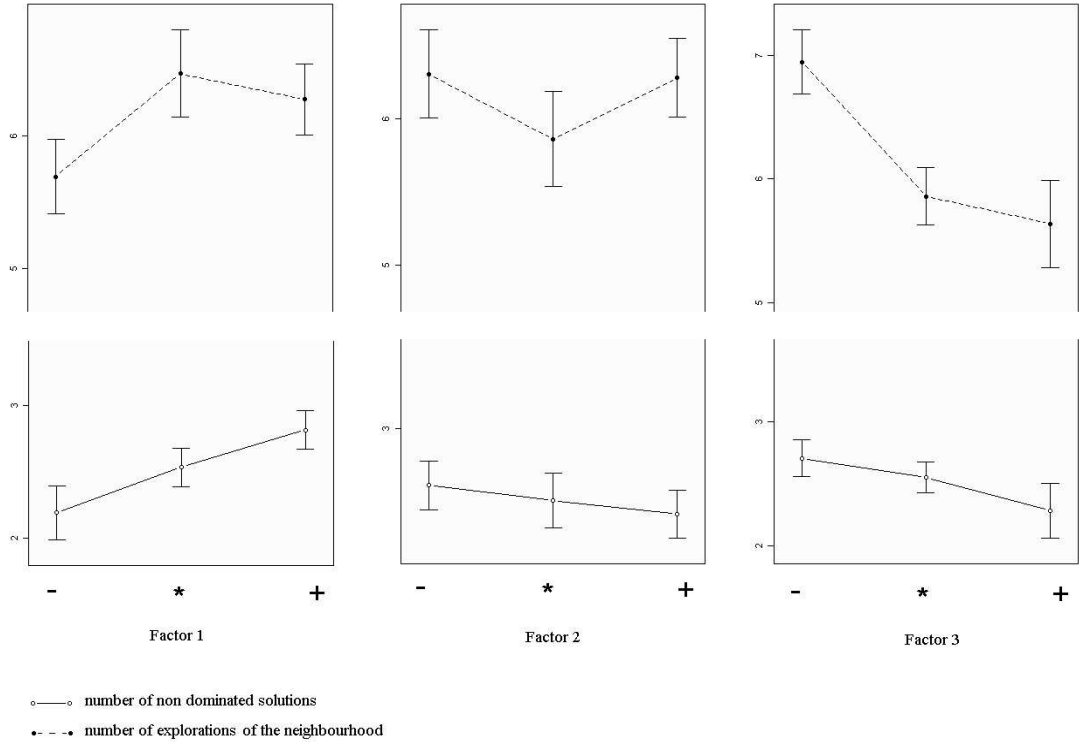}
\caption{\label{fig:Multiobjective_SER} \footnotesize{The numerical results associated to Table \ref{fig:MANOVA_Multiobjective_SER} are shown. The dotted red lines denote the number of non-dominated allocations, whereas the continuous green lines denote the number of of neighborhoods explored.}}
\end{center}
\end{figure}

This experimental result should interpreted as exploratory and aiming to provide clues and suggestions for further analysis about the effect of the initial condition of the system on the outcomes and performance of the SERs. In this respect, the significant effects of dissortative matching advise for the analysis of the dissortative behavior of the economical interaction network.

Any SER intrinsically gives rise to two types of network structures generated by the set of couples of agents interacting along the process:
\small
\begin{itemize}
\item[-] the \emph{between--node--interaction network} (whose edge set is represented by the number of exchanges, that is to say, the number of times a ERP is solved per each couple of agents),
\item[-] the \emph{between--node--flow network} (whose edge set is represented by amount of exchanged commodities for each couple of agents).
\end{itemize}
\normalsize
Both networks can be seen as dynamically changing along the process. Such structures might be statistically analyzed in term of their topological properties. We consider three kind of assortativity measures reflecting the preference for an agent to interact with others that are similar or different in some ways:
\small
\begin{itemize}
\item[-] $Type_1$: couples of agents with highly different marginal utilities are more often commercial partners: $\rho(\delta(c^h, \, c^k ), x_{hk})$;

\item[-] $Type_2$: agents who are more sociable (trade more often) interact frequently with agents who are not sociable: $\rho(\delta(f_h, \, f_k), x_{hk})$--;

\item[-] $Type_3$: the more two agents are different with respect to their marginal utilities, the more they are different with respect to their commercial interactions: $\rho( \delta(c^h, \, c^k ), \delta(f_h, \, f_k ))$.
\end{itemize}
\normalsize
The Greek letter $\delta$ denotes the Euclidean distance, $\rho$ is the Pearson correlation, $x_{hk}$ is the valued of the connection between agent $h$ and $k$ and $f_h$ is the total value of connections of agent $h$, corresponding to the $h^{th}$ row of the AM. The numerical values in Table \ref{tab:Dissortativity} corresponds to the aforementioned assortativities applied to the \emph{interaction network}, corresponding to the instances of Table \ref{tab:SER_BranchCut}.

\begin{table}[H]
\begin{center}
\scalebox{0.83}{
\begin{tabular}{r|rrr|rrr}
  \hline
\multirow{2}{*}{Size} & \multicolumn{3}{c|}{Between--node--flow} & \multicolumn{3}{c}{Between--node--interaction}\\
%\hhline{~------}
     & $Type_1$ & $Type_2$ & $Type_3$ & $Type_1$ & $Type_2$ & $Type_3$\\
  \hline
  10 &  0.40 & 0.48 & 0.63 & 0.70 & 0.67 & 0.74  \\
  10 &  0.46 & 0.66 & 0.61 & 0.85 & 0.63 & 0.74  \\
  10 &  0.60 & 0.48 & 0.75 & 0.71 & 0.70 & 0.75  \\

  15 &  0.47 & 0.31 & 0.62 & 0.74 & 0.48 & 0.56  \\
  15 &  0.33 & 0.36 & 0.58 & 0.58 & 0.44 & 0.67  \\
  15 &  0.24 & 0.48 & 0.53 & 0.56 & 0.74 & 0.66  \\

  20 &  0.28 & 0.41 & 0.61 & 0.39 & 0.62 & 0.54  \\
  20 &  0.23 & 0.18 & 0.46 & 0.54 & 0.48 & 0.55  \\
  20 &  0.12 & 0.06 & 0.37 & 0.48 & 0.45 & 0.42  \\

  25 &  0.14 & 0.18 & 0.39 & 0.55 & 0.66 & 0.53  \\
  25 &  0.36 & 0.32 & 0.60 & 0.65 & 0.56 & 0.66  \\
  25 &  0.14 & 0.17 & 0.51 & 0.48 & 0.70 & 0.49  \\

  30 &  0.09 & 0.08 & 0.40 & 0.42 & 0.55 & 0.53  \\
  30 &  0.24 & 0.20 & 0.67 & 0.56 & 0.62 & 0.68  \\
  30 &  0.26 & 0.33 & 0.60 & 0.61 & 0.63 & 0.65  \\

  35 &  0.11 & 0.29 & 0.40 & 0.44 & 0.59 & 0.43  \\
  35 &  0.14 & 0.28 & 0.50 & 0.46 & 0.55 & 0.48  \\
  35 &  0.14 & 0.26 & 0.49 & 0.46 & 0.58 & 0.53  \\

  40 &  0.25 & 0.22 & 0.53 & 0.44 & 0.64 & 0.58  \\
  40 &  0.28 & 0.23 & 0.58 & 0.68 & 0.52 & 0.64  \\
  40 &  0.26 & 0.18 & 0.69 & 0.64 & 0.64 & 0.60  \\

  45 &  0.23 & 0.30 & 0.55 & 0.62 & 0.60 & 0.54  \\
  45 &  0.29 & 0.24 & 0.61 & 0.57 & 0.59 & 0.58  \\
  45 &  0.21 & 0.21 & 0.63 & 0.58 & 0.57 & 0.61  \\

  50 &  0.08 & 0.28 & 0.36 & 0.35 & 0.55 & 0.32  \\
  50 &  0.16 & 0.32 & 0.41 & 0.45 & 0.62 & 0.42  \\
  50 &  0.24 & 0.17 & 0.60 & 0.51 & 0.50 & 0.65  \\

  55 &  0.14 & 0.53 & 0.17 & 0.39 & 0.52 & 0.48  \\
  55 &  0.17 & 0.33 & 0.38 & 0.29 & 0.53 & 0.44  \\
  55 &  0.19 & 0.37 & 0.38 & 0.47 & 0.56 & 0.43  \\

  60 &  0.35 & 0.45 & 0.60 & 0.54 & 0.57 & 0.62 \\
  60 &  0.20 & 0.30 & 0.43 & 0.34 & 0.50 & 0.52 \\
  60 &  0.16 & 0.38 & 0.29 & 0.39 & 0.51 & 0.48 \\
   \hline
\end{tabular}}
\caption{ \label{tab:Dissortativity}
\footnotesize {Three types of network assortativity.}}
\end{center}
\end{table}

The significative effect of $Fact_3$ (the association between the initial endowment and the marginal utility of the other agent) in the MANOVA of Table \ref{fig:MANOVA_Multiobjective_SER} seems coherent with the $Type_1$ assortativity reported in Table \ref{tab:Dissortativity}, in the vague sense that the difference in the agents marginal utilities is likely to result in high exchange opportunities for agents and, conversely, in many possible convenient allocations (in the sense of Pareto).

Surprisingly, as far as the network corresponding to the between--node--flow is concerned, the $Type_3$ assortativity appear comparatively higher then the others. It might be argued that this is due to the fact that nodes with similar marginal utilities have similar abilities in catching the same exchange opportunities existing in the market. An analogous result is observed for the networks corresponding to the between--node--interaction.

Regarding the $Type_2$ dissortativity of the between--node--interaction, the values in Table \ref{tab:Dissortativity} provide a clear connections with the results of Cook et al. \cite{Cook1983}, who observed that most central nodes (in the sense of eigenvector centrality) were not the most successful in achieving high bargaining power. It can be argued that this achievement relies on his/her connections with poorly connected nodes\footnote{This results contradict most social psychological literature showing that, in experimentally restricted communication networks, the leadership role typically devolves upon the individual in the most central position \cite{Leavitt1951, Berkowitz1956}}, as noted by Bonacich \cite{Bonacich1987}:
\small
\begin{quotation}
\hspace{-6mm} \emph{in bargaining situations, it is advantageous to be connected to those who have few options; power comes from being connected to those who are powerless. Being connected to powerful others who have many potential trading partners reduces one's bargaining power.}
\end{quotation}
\normalsize

Note that the goodness of being connected with powerful or powerless neighbors depends on the type of commodity flowing within the network. If the utility of nodes are related to the amount of obtained information, the non rival nature of information suggests a positive association between the power of a node and the power of its neighbors.

The dissortative behavior of the valued networks generated by the barter process can be probabilistically analyzed using conditionally uniform random network models. For each of the three problems of size $60$ in Table \ref{tab:SER_BranchCut}, the results in Table \ref{tab:SER_CURN_1} show the sample mean and standard deviation of the clustering coefficient and assortativity coefficient of a sample of $20.000$ valued networks with fixed density (summation of the AM components).

\begin{table}[H]
\begin{center}
\scalebox{0.87}{
\begin{tabular}{rrrrrrrr}
\hline
Network                                                           & Property & sample mean & sample std.   & observed value & one tail p-value  & corr CC -- AC  \\
\hline
                                                                  &   CC     &   0.0583    &  0.0099       &  0.0107        &  0.9951           &  0.1075        \\
                                                                  &   AC     &  -0.0181    &  0.0054       & -0.0454        &  0.0000           &                \\
\hhline{~-------}
                                                                  &   CC     &   0.0613    &  0.0114       &  0.0101        &  0.9951           &  -0.0847       \\
\multirow{2}{*}{\rotatebox{90}{Flow}} $\quad$                     &   AC     &  -0.0196    &  0.0056       & -0.0491        &  0.0000           &                \\
\hhline{~-------}
                                                                  &   CC     &   0.0615    &  0.0096       &  0.0390        &  0.9974           &  0.1387        \\
                                                                  &   AC     &  -0.0188    &  0.0058       & -0.0316        &  0.0379           &                \\
\hline
                                                                  &   CC     &   0.0901    &  0.0092       &  0.0822        &  0.7832           &  -0.1125       \\
\multirow{2}{*}{\rotatebox{90}{$\quad$ Interaction}} $\quad$      &   AC     &  -0.0220    &  0.0110       & -0.0454        &  0.0220           &                \\
\hhline{~-------}
                                                                  &   CC     &   0.1085    &  0.0050       &  0.1125        &  0.0992           &   0.1344       \\
                                                                  &   AC     &  -0.0221    &  0.0115       & -0.0491        &  0.0027           &                \\
\hhline{~-------}
                                                                  &   CC     &   0.1125    &  0.0042       &  0.1178        &  0.0576           &  -0.0250       \\
                                                                  &   AC     &  -0.0203    &  0.0128       & -0.0326        &  0.0411           &                \\                              \hline
\end{tabular}}
\caption{\label{tab:SER_CURN_1} \footnotesize Numerical results from the sample obtained with the $q$-kernel method, for each of the six networks associated to the the three barter processes of size $60$ in Table \ref{tab:SER_BranchCut}. The model is based on the conditionally uniform distribution of valued networks with fixed density (summation of the AM components). The sixth column reports the left-tailed p-values.}
\end{center}
\end{table}

Similarly, for the same samples of Table \ref{tab:SER_CURN_1}, the results in Table \ref{tab:SER_CURN_2} show the sample mean and standard deviation of the clustering coefficient and assortativity coefficient of a sample of $10.000$ valued networks with fixed row marginal of the AM generated by the $q$-kernel method.

\begin{table}[H]
\begin{center}
\scalebox{0.87}{
\begin{tabular}{rrrrrrrr}
\hline
Network                                                           & Property & sample mean & sample std.   & observed value & one tail p-value  & corr CC -- AC  \\
\hline
                                                                  &   CC     &   0.0170    &  0.0031       &  0.0107        &  0.9833           &  0.0462        \\
                                                                  &   AC     &  -0.0079    &  0.0053       & -0.0454        &  0.0000           &                \\
\hhline{~-------}
                                                                  &   CC     &   0.0177    &  0.0041       &  0.0101        &  1.0000           &  0.0286        \\
\multirow{2}{*}{\rotatebox{90}{Flow}} $\quad$                     &   AC     &  -0.0064    &  0.0060       & -0.0491        &  0.0000           &                \\
\hhline{~-------}
                                                                  &   CC     &   0.0433    &  0.0072       &  0.0390        &  0.7895           &  -0.0462       \\
                                                                  &   AC     &  -0.0144    &  0.0073       & -0.0316        &  0.0092           &                \\
\hline
                                                                  &   CC     &   0.0515    &  0.0182       &  0.0822        &  0.1179           &  0.0067        \\
\multirow{2}{*}{\rotatebox{90}{$\quad$ Interaction}} $\quad$      &   AC     &  -0.0251    &  0.0102       & -0.0454        &  0.0339           &                \\
\hhline{~-------}
                                                                  &   CC     &   0.0848    &  0.0168       &  0.1125        &  0.0870           &  -0.1542       \\
                                                                  &   AC     &  -0.0173    &  0.0143       & -0.0491        &  0.0254           &                \\
\hhline{~-------}
                                                                  &   CC     &   0.0633    &  0.0169       &  0.6384        &  0.0332           &  0.0932        \\
                                                                  &   AC     &  -0.0154    &  0.0101       & -0.4786        &  0.0433           &                \\                              \hline
\end{tabular}}
\caption{\label{tab:SER_CURN_2} \footnotesize Numerical results from the sample obtained with the $q$-kernel method, for each of the six networks associated to the the three barter processes of size $60$ in Table \ref{tab:SER_BranchCut}. The model is based on the conditionally uniform distribution of valued networks with row marginal density. The sixth column reports the left-tailed p-values.}
\end{center}
\end{table}

The results in tables \ref{tab:SER_CURN_1} and \ref{tab:SER_CURN_2} are quite confirmatory, as the negative values of the CC and AC between row marginal can not be explained based on the supposed conditional randomness.

In a series of computational experiments Kang \cite{Kang2007} showed an interesting relationship between the variation at the individual level of a network and its assortative behavior. He found that when actors are connected with similarly central alters, the overall variation at the individual centralities (network centralization) is low.

The global picture emerging from the observed computational results strongly supports the previously discussed micro–macro linkages. This is particularly true when the dissortative pattern and the network centralization are taken into account \cite{Leavitt1951, Berkowitz1956, Cook1983, Kang2007}. Indeed, this strategic model of network formation is capable of internalizing many and varied assumption on agent behavior, allowing to test hypothesis on the arising of open and closed network structures from the economical interaction.

\section{Summary and future directions}\

We studied the use of barter processes for solving problems of bargaining on a discrete set, representing markets with indivisible goods and fixed exogenous prices. We showed that the allocation space is characterized by a block diagonal system of linear constraints, whose structural properties might be exploited in the construction and analysis of barter processes. Using Proposition \ref{prop:2} and the characterization of the space of integer solutions of the ERP, we were able to derive a constructive procedure to obtain its Pareto frontier, as shown by Corollary \ref{corol:1} and Corollary \ref{corol:2}.

Further research on this topic should include the characterization of the integer points in the null space of a general reallocation problem with fixed prices to obtain a closed form solution of a general problem of reallocating integer amounts of $m$ commodities among $n$ agents with fixed prices.

An open problem, which has not been investigated in this paper, is the formulation of equilibrium conditions for this rationing scheme proposed in Section 3, as suggested by Dreze \cite{Dreze1975} for the case of continuous allocation space.

In Section 4 we proposed a mathematical programming model for the problem of reallocating integer amounts of $m$ commodities among $n$ agents with fixed prices on a sparse network structure with nodal capacities. Further research on this issue should include a mathematical properties of a SER in dealing with markets with sparsely connected agents, as formulated in (\ref{eq:cons7}).

\end{document}